\documentclass{article}
\usepackage{fullpage}
\usepackage{lmodern}
\usepackage{microtype,xspace,multirow}
\usepackage{graphicx}
\usepackage{plaatjes}
\graphicspath{{./figures/}}
\usepackage{amsfonts,amsmath,amsthm,amssymb}
\usepackage{color}
\usepackage[hidelinks]{hyperref}

\usepackage[backend=bibtex,style=numeric-comp,sorting=none,firstinits=true,doi=false,url=false,maxbibnames=99]{biblatex}
\AtEveryBibitem{
 \clearlist{address}
 \clearfield{date}
 \clearfield{eprint}
 \clearfield{isbn}
 \clearfield{issn}
 \clearlist{location}
 \clearfield{month}
 \clearfield{series}
 \clearfield{pages}

 \ifentrytype{book}{}{
  \clearlist{publisher}
  \clearname{editor}
 }
}

\let\emptyset\varnothing

\def\nicefrac#1#2{
	\raise.5ex\hbox{$#1$}%
	\kern-.1em/\kern-.15em%
	\lower.25ex\hbox{$#2$}}

\newcommand{\Sp}{\mathcal{S}}
\newcommand{\R}{\mathbb{R}}
\newcommand{\Li}{\mathcal{L}}
\newcommand\e\emph
\newcommand\eps{\ensuremath{\varepsilon}\xspace}
\newcommand\s{\ensuremath{s}\xspace}
\renewcommand\t{\ensuremath{t}\xspace}
\newcommand\T{\ensuremath{\mathbb{T}}\xspace}
\newcommand{\opt}{\ensuremath{\mathsf{opt}}\xspace}
\newcommand{\apx}{\ensuremath{\mathsf{apx}}\xspace}
\newcommand{\apxe}{\ensuremath{\apx_e}\xspace}

\newcommand\dr{\textsf{DiffuseReflection}\xspace}
\newcommand\drs{\textsf{DiffRefl}\xspace}
\newcommand\mlp[2]{\textsf{MinLinkPath}\ensuremath{_{#1,#2}}\xspace}
\newcommand\ske[1]{\ensuremath{D|^{#1}}\xspace}

\newtheorem{theorem}{Theorem}
\newtheorem{lemma}[theorem]{Lemma}
\newtheorem{corollary}[theorem]{Corollary}
\newtheorem{proposition}{Proposition}
\newtheorem{definition}{Definition}

\newtheorem{remark}{Remark}

\title{On the complexity of minimum-link path problems\footnote{An abridged version of this paper appeared in the proceedings of the 32nd International Symposium on Computational Geometry in 2016.}}

\author{
Irina Kostitsyna\thanks{Dept. of Mathematics and Computer Science, TU Eindhoven, Eindhoven, the Netherlands}
\and
Maarten L\"offler\thanks{Dept. of Computing and Information Sciences, Utrecht University, Utrecht, the Netherlands}
\and
Valentin Polishchuk\thanks{Communications and Transport Systems, ITN, Link\"oping University, Sweden}
\and
Frank Staals\thanks{MADALGO, Aarhus University, Aarhus, Denmark}
}

\date{}

\newcommand{\mkmcal}[1]{\ensuremath{\mathcal{#1}}\xspace}

\newcommand{\E}{\mkmcal{E}}

\newcommand{\myremark}[3]{\textcolor{blue}{\textsc{#1 #2:}} \textcolor{red}{\textsf{#3}}}
\renewcommand{\myremark}[3]{}

\newcommand{\frank}[2][says]{\myremark{Frank}{#1}{#2}}
\newcommand{\irina}[2][says]{\myremark{Irina}{#1}{#2}}
\newcommand{\maarten}[2][says]{\myremark{Maarten}{#1}{#2}}
\newcommand{\val}[2][says]{\myremark{Val}{#1}{#2}}

\newcommand{\old}[1]{}

\begin{document}


\maketitle

\begin{abstract}
We revisit the minimum-link path problem: Given a polyhedral domain and two points in it, connect the points by a polygonal path with minimum number of edges. We consider settings where the vertices and/or the edges of the path are restricted to lie on the boundary of the domain, or can be in its interior.
Our results include bit complexity bounds, a novel general hardness construction, and a polynomial-time approximation scheme.
We fully characterize the situation in 2 dimensions, and provide first results in dimensions 3 and higher for several variants of the problem.

  Concretely, our results resolve several open problems.  We prove that
  computing the minimum-link \e{diffuse reflection path}, motivated by ray
  tracing in computer graphics, is NP-hard, even for two-dimensional polygonal
  domains with holes.  This has remained an open problem~\cite{diffuse} despite
  a large body of work on the topic. We also resolve the open problem from~\cite{mrw} mentioned
  in the handbook~\cite{handbook04} (see Chapter~27.5, Open problem~3) and The
  Open Problems Project~\cite{topp} (see Problem~22): ``What is the complexity
  of the minimum-link path problem in 3-space?'' Our results imply that the
  problem is NP-hard even on terrains (and hence, due to discreteness of the answer, there is no
  FPTAS unless P=NP), but admits a PTAS.
\end{abstract}


\clearpage

\tableofcontents

\clearpage

\section{Introduction}
\label{sec:introduction}

The minimum-link path problem is fundamental in computational geometry~\cite{suri,ghosh,homot,kahan,mrw,revisited,n9,Guibas:1991:APS:648003.743125}.
It concerns the following question: given a polyhedral domain $D$ and two points $s$ and $t$ in $D$, what is the polygonal path connecting $s$ to $t$ that lies in $D$ and has as few links as possible?

In this paper, we revisit the problem in a general setting which encompasses several specific variants that have been considered in the literature.
First, we nuance and tighten results on the bit complexity involved in optimal minimum-link paths. Second, we present and apply a novel generic NP-hardness construction. Third, we extend a simple polynomial-time approximation scheme.

Concretely, our results resolve several open problems.
We prove that computing the minimum-link \e{diffuse reflection path} in polygons with holes~\cite{diffuse} is NP-hard, and we prove that the minimum-link path problem in 3-space~\cite{handbook04} (Chapter~27.5, Open problem~3) is NP-hard (even for terrains). In both cases, there is no FPTAS unless P=NP, but there is a PTAS.

We use terms \e{links} and \e{bends} for edges and vertices of the path, saving the terms \e{edges} and \e{vertices} for those of the domain
(also historically, minimum-link paths used to be called \e{minimum-bend} \cite{bends1,bends2,bends3}).

\subsection {Problem Statement, Domains and Constraints}
\label{sec:envs}

Due to their diverse applications, many different variants of minimum-link paths have been considered in the literature.
These variants can be categorized by two aspects.
Firstly, the \emph {domain} can take very different forms.
We select several common domains, ranging from a simple polygon in 2D to complex scenes in full 3D or even in higher dimensions.
Secondly, the links and bends of the solution paths are sometimes \emph {constrained} to lie on the boundary of the domain, or bends may be restricted to vertices or edges of the domain.
We now survey these settings in more detail.

\paragraph {Problem Statement}

Let $D$ be a closed connected $d$-dimensional polyhedral domain. For $0 \le a \le d$ we denote by \ske{a} the $a$-skeleton of $D$; that is, its $a$-dimensional subcomplex. For instance, \ske{d-1} is the boundary of $D$; \ske0 is the set of vertices of $D$. Note that \ske{a} is not necessarily connected.

\begin {definition}
We define $\mlp{a}{b}(D, s, t)$, for $0 \le a \le b \le d$ and $1 \le b$,
to be the problem of finding a minimum-link polygonal path in $D$
between two given points $s$ and $t$, where the bends of the solution (and $s$ and $t$) are restricted to lie in \ske a and the links of the solution are restricted to lie in \ske b.
\end {definition}
Fig.~\ref {fig:model-example} illustrates several instances of the problem in different domains.
\eenplaatje {model-example} {Left: \mlp22 in a polygon with holes. Middle: \mlp12 on a polyhedron. Right: \mlp03 on a polyhedral terrain.}

\paragraph {Domains}

We recap the various settings that have been singled out for studies in computational geometry.
We remark that we will not survey the rich field of path planning in rectilinear, or more generally, $C$-oriented worlds \cite{aedegeest}; all our paths will be assumed to be unrestricted in terms of orientations of their links.

One classical distinction between working setups in 2D is \e{simple polygons}
vs.\ \e{polygonal domains}. The former are a special case of the latter: simple
polygons are domains without holes. Many problems admit more efficient
solutions in simple polygons---loosely speaking, the golden standard is running
time of $O(n)$ for simple polygons and of $O(n\log n)$ for polygonal domains of
complexity $n$. This is the case, e.g., for the shortest path problem
\cite{ghlst,hs}. For minimum-link paths,
$O(n)$-time algorithms are known for simple polygons 
\cite{suri,ghosh,homot}, but for polygonal domains with holes the fastest known algorithm runs
in nearly quadratic time \cite{mrw}, which may be close to optimal due to
3SUM-hardness of the problem \cite{revisited}.
Even more striking is the
difference in the watchman route problem (find a shortest path to see all of
the domain), which combines path planning with visibility: in simple polygons
the optimal route can be found in polynomial time \cite{touring,bengt} while
for domains with holes the problem cannot be approximated to within a
logarithmic factor unless P=NP \cite{watchman}. Finding minimum-link watchman
route is NP-hard even for simple polygons \cite{alsuwaiyel}.\frank{I think we
  can drop the part about the watchman problem. Yes, it is somewhat related,
  but (i) I think it is generally accepted that min-link path is interesting,
  (ii) there are already plenty of references on that, (iii) we already have
  more than two pages of references..... } \maarten {Let's delete it after branching the SoCG version and keep it for the full version.}

In 3D, a \e{terrain} is a polyhedral surface (often restricted to a bounded region in the $xy$-projection) that is intersected only once by any vertical line. Terrains are traditionally studied in GIS applications and are ubiquitous in computational geometry~\cite {removal,new3d}.\val[wonders]{why (these) citations? remove them?}
Minimum-link paths are closely related to visibility problems, which have been studied extensively on terrains~\cite{terrainsurvey,hierarchy,partgard,apxvis,simplification,multiple}.
One step up from terrains, we may consider \e{simple} polyhedra (surfaces of genus $0$), or \e{full 3D} scenes.
Visibility has been studied in full 3D as well~\cite {moet2008,silhouette,faceguarding}.
To our knowledge, minimum-link paths in higher dimensions have not been studied before (with the exception of~\cite{recworlds} that considered rectilinear paths).

\paragraph {Constraints}

In path planning on polyhedral surfaces or terrains, it is standard to restrict paths to the terrain.
Minimum-link paths, on the other hand, have various geographic applications, ranging from feature simplification~\cite {Guibas:1991:APS:648003.743125} to visibility in terrains~\cite{terrainsurvey}.
In some of these applications, paths are allowed to live in free space, while bends are still restricted to the terrain.
In the GIS literature, out of simplicity and/or efficiency concerns, it is common to constrain bends even further to vertices of the domain (or, even more severely, the terrain itself may restrict vertices to grid points, as in the popular \emph{digital elevation map} (DEM) model).

In a vanilla min-link path problem the location of vertices (bends) of the path are unconstrained, i.e., they can occur anywhere in the free space. In the \e{diffuse reflection} model
\cite{diffuse,diffuseexp,constrained,aronov1,aronovk,n9}
 the bends are restricted to occur on the boundary of the domain. Studying this kind of paths is motivated by ray tracing in realistic rendering of 3D scenes in graphics, as light sources that can reach a pixel with fewer reflections make higher contributions to intensity of the pixel \cite{graphicsBook,removal}. Despite the 3D graphics motivation, all work on diffuse reflection has been confined to 2D polygonal domains, where the path bends are restricted to edges of the domain.

\subsection {Representation and Computation}

In computational geometry, the standard model of computation is the \emph {Real RAM},
which represents data as an infinite sequence of storage cells which can store any real number
or integer.
\maarten [notes for future reference] {Maybe uncomment more details for full version?}\irina{uncommented}
The model supports standard operations (such as addition, multiplication, or taking square-roots) in constant time.
The Real RAM is preferred for its elegance, but may not always be the best representation of physical computers.
For example, the \emph{floor} function is often allowed, which can be used to truncate a
real number to the nearest integer, but points at a flaw in the model:
if we were allowed to use it arbitrarily, the Real RAM could solve PSPACE-complete
problems in polynomial time~\cite{Schoenhage79}.
In contrast, the \emph {word RAM} stores a sequence of $w$-bit words, where $w \geq \log n$ (and $n$ is the problem size).
Data can be accessed arbitrarily, and standard
operations, such as Boolean operations
(\texttt{and}, \texttt{xor}, \texttt{shl}, $\ldots$), addition, or
multiplication take constant time. There are many variants of
the word RAM, depending on precisely which instructions are
supported in constant time. The general consensus seems
to be that any function in $\text{AC}^0$
is acceptable.\footnote{$\text{AC}^0$ is the
class of all functions $f: \{0,1\}^* \rightarrow \{0,1\}^*$ that
can be computed by a family of circuits $(C_n)_{n \in \mathbb{N}}$ with the
following properties: (i) each $C_n$ has $n$ inputs; (ii) there exist constants
$a,b$, such that $C_n$ has at most $an^b$ gates, for $n\in \mathbb{N}$;
(iii) there is a constant $d$ such that for all $n$ the length of the longest
path from an input to an output in $C_n$ is at most $d$ (i.e., the
circuit family has bounded depth); (iv) each gate
has an arbitrary number of incoming edges (i.e., the \emph{fan-in} is
unbounded).} However, it is always preferable to rely on a set of operations
as small, and as non-exotic, as possible.
Note that multiplication is not in $\text{AC}^0$~\cite{FurstSaSi84}.
Nevertheless, it is usually
included in the word RAM instruction set~\cite{FredmanWi94}.
The word RAM is much closer to reality, but complicates the analysis of geometric problems.

\val[edited this paragraph]{put in stuff about shortest paths, but removed a lot of other things (maybe shouldn't have removed all of them...)}In many cases, the difference is unimportant, as the real numbers involved in solving geometric problems are in fact algebraic numbers of low degree in a bounded domain, which can be described exactly with constantly many words.
Path planning is notoriously different in this respect. Indeed, in the Real RAM both the Euclidean shortest paths and the minimum-link paths in 2D can be found in optimal times. On the contrary, much less is known about the complexity of the problems in other models. For $L_2$-shortest paths the issue is that their length is represented by the sum of square roots and it is not known whether comparing the sum to a number can be done efficiently (if yes, one may hope that the difference between the models vanishes). Slightly more is known about minimum-link paths, for which the models are \e{provably} different: Snoeyink and Kahan~\cite{kahan} observed that the region of points reachable by $k$-link paths may have vertices needing $\Omega(k\log n)$ bits to describe. One of the results in this paper is the matching upper bound on the bit complexity of min-link paths.\val[hesitates]{Is speaking about results out of place here?}\irina{because it's one of the ``closing gaps'' results, I think it is appropriate mention I here} 


Relatedly, when studying the computational complexity of geometric problems, it is often not trivial to show a problem is in NP. Even if a potential solution can be verified in polynomial time, if such a solution requires real numbers that cannot be described succinctly, the set of solutions to try may be too large.
Recently, there has been some interest in computational geometry in showing problems are in NP~\cite{packinnp} (see also~\cite{string}).

A common practical approach to avoiding bit complexity issues is to approximate the problem by restricting solutions to use only vertices of the input. In minimum-link paths, this corresponds to \mlp0b. Although such paths can be computed efficiently,
a simple example (Appendix~\ref {ap:ve}) shows that the number of links in such a setting may be a linear factor higher than when considering geometric versions.

\subsection {Results}
\label{sec:results}
We give hardness results and approximation algorithms for various versions of the min-link path problem. Specifically,
\begin{itemize}
\item \val[suggests to rephrase]{In Section~\ref{sec:bit2d} we give an $\Omega(n \log n)$ lower bound on the bit complexity of some bends of min-link paths in 2D.}
In Section~\ref{sec:bit2d} we show a general lower bound on the bit complexity of min-link paths of $\Omega(n \log n)$ bits for some coordinates. (This was previously claimed, but not proven, by Snoeyink and Kahan~\cite{kahan}.) We show that the bound is tight in 2D
and we argue that this implies that \mlp{a}{2} is in NP.
In Section~\ref{sec:bit3d}, we argue that in 3D the boundary of the $k$-illuminated region can consist of $k$-th order algebraic curves, potentially leading to exponential bit complexity.
\item In Section~\ref{sub:blueprint} we present a blueprint for showing NP-hardness of minimum link problems. We apply it to prove NP-hardness of the diffuse reflection path problem (\mlp12) in 2D polygonal domains with holes in Section~\ref{sub:mlp12}.
In~Section~\ref{sec:hardness3d}, we use the same blueprint to prove that all non-trivial versions, defined above, of min-link problems in 3D are weakly NP-hard.
We also note that the min-link problems have no FPTAS and no additive approximation (unless P=NP).
\item In Section~\ref{sec:apxs2d} we extend the 2-approximation algorithm from \cite[Ch.~27.5]{handbook04}, based on computing weak visibility between sets of potential locations of the path's bends, to provide a simple PTAS for \mlp22, which we also adapt to \mlp12.
In Section~\ref{sec:apxs3d} we give simple constant-factor approximation algorithms for higher-dimensional minimum-link path versions, which can then be used in the same way to show that all versions admit PTASes.
\item In Section~\ref{sec:vis} we focus on \mlp23 (diffuse reflection in 3D) on terrains---the version that is most important in practice.\val{Misleading: we said diffrefl is important in graphics, but diffrefl in terrains is important for other reasons} We give a 2-approximation algorithm that runs faster than the generic algorithm from \cite[Ch.~27.5]{handbook04}. We also present an $O(n^4)$-size data structure encoding visibility between points on a terrain and argue that the size of the structure is asymptotically optimal.
\maarten {Don't forget to drop this bullet point if we drop the section for SoCG.}
\end{itemize}

\maarten {Cite \cite {faceguarding} in section on 3D complexity?}

\maarten {Mention somewhere that in 3D, when restricting vertices to grid points (low complexity points), we get arbirarily more links).}

Our results are charted and compared to existing results in Table~\ref {tab:results}.

\begin{table}[t]
\begin{tabular}[b]{l|p{2cm}|p{5cm}|p{5cm}}
\mlp{a}{b}
&$b=1$&$b=2$&$b=3$\\
\hline
$a=0$ & $O(n)$
      & $O(n^2)$
      & $O(n^2)$ \\
\hline
$a=1$ & $O(n)$
      & Simple Polygon: $O(n^9)$\cite{n9}\newline
        Full 2D: NP-hard$\star$\newline
        PTAS$\star$
      & NP-hard$\star$ (even in terrains)\newline
        PTAS$\star$
      \\
\hline
$a=2$ & N/A
      & Simple Polygon: $O(n)$\cite{suri}\newline
        Full 2D: $O(n^2\alpha(n)\log^2{n})$\cite{mrw}\newline
        PTAS$\star$
      & NP-hard$\star$ (even in terrains)\newline
        PTAS$\star$
      \\
\hline
$a=3$ & N/A
      & N/A
      & Terrains: $O(1)$\newline
        Full 3D: NP-hard$\star$\newline
        PTAS$\star$\\
\end{tabular}
\caption{Computational complexity of \mlp{a}{b} for $a \le b \le 3$. Results with citations are known, results marked with $\star$ are from this paper. Results without marks are trivial.}
\label{tab:results}
\end{table}
\val[finds]{$a=0$ version non-interesting (in the earlier draft it was ignored altogether, except for the small appendix)}
\maarten{I think it makes sense to include them even if it is non-interesting, since it is still part of the overall picture. However, I'm not sure the $O(|n^2)$ is optimal for \mlp02, and we haven't really though about it. So that might be a reason to leave out the row...}





\section {Algebraic Complexity in \texorpdfstring{$\R^2$}{R2}}
\label{sec:bit2d}

\subsection{Lower bound on the Bit complexity}
\label{sec:bit2d-low}

Snoeyink and Kahan~\cite {kahan} claim to ``give a simple instance in which
representing path vertices with rational coordinates requires $\Theta(n^2 \log
n)$ bits''. In fact, they show that the boundary of the region reachable from $s$ (a point with integer coordinates specified with $O(\log n)$ bits)
with $k$ links may have vertices whose coordinates have bit complexity $k \log n$. Note however, that this does not directly imply that a
minimum-link path from $s$ to another point \t with low-complexity (integer) coordinates must necessarily have such high-complexity bends (i.e., if \t itself is not a high-complexity vertex of a $k$-reachable region, one potentially could hope to avoid also placing the internal vertices of a min-link path to \t on such high-complexity points). Below we present a construction where the intermediate vertices must actually use $\Omega(k \log n)$ bits to be described, even if $s$ and $t$ can be specified using only $\log n$ bits each. We first prove this for the \mlp12 variant of the problem, and then extend our results to paths that may bend anywhere within the polygon, i.e. \mlp22.

\begin{lemma}\label{lem:bit_complexity_lower_bound_R2_edge}
There exists a simple polygon $P$, and points \s and \t in $P$ such that: (i) all the coordinates of the vertices of $P$ and of \s and \t can be represented using $O(\log n)$ bits, and (ii) any \s-\t min-link path that bends only on the edges of $P$ has vertices whose coordinates require $\Omega(k\log n)$ bits, where $k$ is the length of a min-link path between \s and \t.
\end{lemma}

\begin{proof}
  We will refer to numbers with $O(\log n)$ bits as \emph{low-complexity}.

  \eenplaatje {bits-spiral} {(a) A spiral, as used in the construction by Kahan
    and Snoeyink. It uses integer coordinates with $O(\log n)$ bits. (b) The
    general idea. \frank{we may want to show the restricted intervals $J_i$ in
      such a fig as well.}}

  The general idea in our construction is as follows. We start with a
  low-complexity point $s' = b_0$ on an edge $e_0$ of the polygon. We then
  consider the furthest point $b_{i+1}$ on the boundary of $P$ that is
  reachable from $b_i$. More specifically, we require that any point on the
  boundary of $P$ between $s'$ and $b_i$ is reachable by a path of at most $i$
  links. We will obtain $b_{i+1}$ by projecting $b_i$ through a vertex
  $c_i$. Each such a step will increase the required number of bits for
  $b_{i+1}$ by $\Theta(\log n)$. Eventually, this yields a point $b_k$ on edge
  $e_k$. Let $t'$ be the $k$-reachable point on $e_k$ closest to $b_k$ that has
  low complexity. Since all points along the boundary from $s'$ to $b_k$ are
  reachable, and the vertices of $P$ have low complexity, such a point is
  guaranteed to exist. We set $a_k = t'$ and project $a_i$ through $c_{i-1}$ to
  $a_{i-1}$ to give us the furthest point (from $t'$) reachable by $k-i$
  links. See Fig.~\ref{fig:bits-spiral} for an illustration.

  The points in the interval $I_i = [a_i,b_i]$, with $1 \leq i < k$, are
  reachable from $s'$ by exactly $i$ links, and reachable from $t'$ by exactly
  $k-i$ links. So, to get from $s'$ to $t'$ with $k$ links, we need to choose
  the $i^\textrm{th}$ bend of the path to be within the interval
  $[a_i,b_i]$. By construction, the intervals for $i$ close to one or close to
  $k$ must contain low-complexity points. We now argue that we can build the
  construction in such a way that $I_{k/2}$ contains no low-complexity
  points.

  Observe that, if an interval contains no points that can be described with
  fewer than $m$ bits, its length can be at most $2^{-m}$. So, we have to show
  that $I_{k/2}$ has length at most $2^{-k\log n}$. 

  \eenplaatje {bits-shrink} {The interval $I_i$ of length $w_i$ produces an
    interval $I_{i+1}$ of length at most $w_{i+1} = h_i/\Theta(n) =
    \Theta(w_i/n^2)$, where $h_i = w_i/(w_i + \Theta(n))$. When the
    $i^{\mathrm{th}}$ link can be anywhere in region $R_i$ (shown in yellow),
    it follows that $R_i$ has height at most $h_i$, and width at most $w_i$.}

  By construction, the interval $I_k$ has length at most one. Similarly, the
  length of $I_0$ can be chosen to be at most one (if it is larger, we can
  adjust $s'=b_0$ to be the closest integer point to $a_0$). Now observe that
  that in every step, we can reduce the length $w_i$ of the interval $I_i$ by a
  factor $\Theta(n^2)$, using a construction like in
  Fig.~\ref{fig:bits-shrink}. Our overall construction is then shown in
  Fig.~\ref{fig:bits-cherry}.

  \eenplaatje {bits-cherry} {An overview of our polygon $P$ and the
    minimum-link path that has high-complexity coordinates.}

  It follows that $I_{k/2}$ cannot contain two low-complexity points that are
  close to each other. Note however, that it may still contain one such a
  point. However, it is easy to see that there is a sub-interval $J_{k/2} =
  [\ell_{k/2},r_{k/2}] \subseteq I_{k/2}$ of length $w_{k/2}/2$ that contains
  no points with fewer than $k \log n$ bits. By choosing $J_{k/2}$ we have
  restricted the interval that must contain the $(k/2)^{\mathrm{th}}$
  bend. This also restricts the possible positions for the $i^\mathrm{th}$ bend
  to an interval $J_i \subseteq I_i$. We find these intervals by projecting
  $\ell_{k/2}$ and $r_{k/2}$ through the vertices of $P$. Note that $s'$ and
  $t'$ may not be contained in $J_0$ and $J_k$, respectively, so we pick a new
  start point $s \in J_0$ and en point $t \in J_k$ as follows. Let $m_{k/2}$ be
  the mid point of $J_{k/2}$ and project $m_i$ through the vertices of $P$. Now
  choose $s$ to be a low-complexity point in the interval $[m_0,r_0]$, and $t$
  to be a low-complexity point in the interval $[\ell_k,m_k]$. Observe that
  $[m_0,r_0]$ and $[\ell_k,m_k]$ have length $\Theta(1)$---as
  $[\ell_{k/2},m_{k/2}]$ and $[m,_{k/2},r_{k/2}]$ have length $w_{k/2}/4$---and
  thus contain low complexity points. Furthermore, observe that $t$ is indeed
  reachable from $s$ by a path with $k-1$ bends (and thus $k$ links), all of
  which much lie in the intervals $J_i$, $1 \leq i < k$. For example using the
  path that uses all points $m_i$. Thus, we have that $t$ is reachable from $s$
  by a minimum-link path of $k$ links, and we need $\Omega(k\log n)$ bits to
  describe the coordinates of the vertices in such a path.
\end{proof}

\begin{lemma}\label{lem:bit_complexity_lower_bound_R2}
There exists a simple polygon $P$, and points \s and \t in $P$ such that: (i) all the coordinates of the vertices of $P$ and of \s and \t can be represented using $O(\log n)$ bits, and (ii) any \s-\t min-link path has vertices whose coordinates require $\Omega(k\log n)$ bits, where $k$ is the length of a min-link path between \s and \t.
\end{lemma}

\begin{proof}
  We extend the construction from
  Lemma~\ref{lem:bit_complexity_lower_bound_R2_edge} to the case in which the
  bends may also lie in the interior of $P$. Let $B_i$ denote the region in $P$
  that is reachable from $s'$ by exactly $i$ links, let $A_i$ the region
  reachable from $t'$ by exactly $k-i$ links, and let $R_i = B_i \cap A_i$. To
  get from $s'$ to $t'$ with $k$ links, the $i^\textrm{th}$ bend has to lie in
  $R_i$. Now observe that this region is triangular, and incident to the
  interval $I_i$ (see e.g. Fig.~\ref{fig:bits-shrink} for an
  illustration). This region $R_i$ has width at most $w_i$ and height at most
  $h_i = w_i/(w_i + \Theta(n))$. Therefore, we can again argue that $R_{k/2}$
  is small, and thus contains at most one low-complexity point $p$. We then
  again choose a region $R'_{k/2} \subseteq R_{k/2}$ of diameter $w_{k/2}/2$
  that avoids point $p$. The remainder of the argument is analogous to the one
  before; we can pick points $s$ and $t$ in the restricted regions $R'_0$ and
  $R'_k$ that are reachable by a minimum-link path of $k-1$ bends, all of which
  have to lie in the regions $R'_i$. It follows that we again need
  $\Omega(k\log n)$ bits to describe the coordinates of the vertices in such a
  path.
\end{proof}

\subsection{Upper bound on the Bit complexity}
\label{sec:bit2d-up}

We now show that the bound of Snoeyink and Kahan~\cite{kahan} on the complexity of $k$-link reachable regions is tight: representing the regions $\mathcal{R}$ as
polygons with rational coordinates requires $O(n^2 \log n)$ for any polygon
$P$, assuming that representation of the coordinates of any vertex of $P$
requires at most $c_0\log n$ bits for some constant $c_0$. Thus, we have a
matching lower and upper bound on the bit complexity of a minimum-link path in
$\R^2$.

Consider a simple polygon $P$ with $n$ vertices, and a point $s\in P$. Analogous to \cite{kahan}, define a sequence of regions $\mathcal{R}=\{R_1,R_2,R_3,\dots\}$, where $R_1$ is a set of all points in $P$ that see $s$, and $R_{i+1}$ is a region of points in $P$ that see some point in $R_i$ for $i\geq 1$. In other words, region $R_{i+1}$ consists of all the points of $P$ that are illuminated by region $R_i$.

\paragraph{Construction of region $R_{i+1}$.} If $P$ is a simple polygon, then
$R_{i+1}$ is also a simple polygon, consisting of $O(n)$ vertices. We will
bound the bit complexity of a single vertex of $R_{i+1}$. The vertices of such
a region are either
\begin{itemize}
	\item original vertices of $P$,
	\item intersection points of $P$'s boundary with lines going through reflex
      vertices of $P$, or
	\item intersection points of $P$'s boundary with rays emanating from the vertices of $R_i$ and going through reflex vertices of $P$.
\frank{I don't see how these types differ from the type 2 ones? (or at least, if you define $R_0 = \{s\}$ then I think they should be the same, right?} \irina{the bit complexity will be different}
\end{itemize}

Only the last type of vertices can lead to an increase in bit
complexity. 
Each of these vertices is defined as an intersection point of two lines: one of
the lines passes through two vertices of $P$, say $a=(x_a,y_a)$ and
$b=(x_b,y_b)$, and, therefore, has a $O(\log n)$ bit representation. The other
line passes through one vertex of $P$, say $c=(x_c,y_c)$, with coordinates of
$O(\log n)$ bit complexity, and one vertex of region $R_i$, say $d=(x_d,y_d)$,
with coordinates of potentially higher complexity. The coordinates of the
intersection can then be calculated by the following formula:
\begin{equation}
\begin{pmatrix}
x^*\\
y^*
\end{pmatrix}=
\begin{pmatrix}
\dfrac{(x_b y_a-x_a y_b+x_a y_c-x_b y_c)x_d+(x_b x_c-x_a x_c)y_d+x_a y_b x_c-y_a x_b x_c}{(y_a-y_b)x_d-(x_a-x_b)y_d+x_a y_c-y_a x_c-x_b y_c+y_b x_c}\\[12pt]
\dfrac{(y_a y_c-y_b y_c)x_d + (x_b y_a-x_c y_a-x_a y_b+x_c y_b) y_d+x_a y_b y_c-y_a x_b y_c}{(y_a-y_b)x_d-(x_a-x_b)y_d+x_a y_c-y_a x_c-x_b y_c+y_b x_c}
\end{pmatrix}=
\begin{pmatrix}
\dfrac{A'_1 x_d+B'_1 y_d+C'_1}{E' x_d+F' y_d+G'}\\[12pt]
\dfrac{A'_2 x_d+B'_2 y_d+C'_2}{E' x_d+F' y_d+G'}
\end{pmatrix}\,.
\label{eq:intersection-point}
\end{equation}
Point $d$ lies on the boundary of $P$. Denote the end points of the side it belongs to as $u$ and $v$. Then the following relation between the coordinates of $d$ holds:
\[
y_d=\frac{(y_u-y_v)x_d+x_u y_v-y_u x_v}{x_u-x_v}\,.
\]
Thus, Equation~\ref{eq:intersection-point} can be rewritten as:
\begin{equation}
\begin{pmatrix}
x^*\\
y^*
\end{pmatrix}=
\begin{pmatrix}
\dfrac{A_1 x_d+B_1}{C x_d+D}\\[12pt]
\dfrac{A_2 x_d+B_2}{C x_d+D}
\end{pmatrix}\,,
\label{eq:intersection-point-new}
\end{equation}
where each of $A_1$, $A_2$, $B_1$, $B_2$, $C$, and $D$ has bit complexity not greater than $c \log n$ for some constant $c$ (here, it is enough to choose $c=4 c_0$). Let $x_d$ be represented as a rational number $\nicefrac{p}{q}$, where $p$ and $q$ are mutually prime integers. Then the number of bits required to represent $x_d$ is $sp(x_d)=\lceil\log (p+1)\rceil+\lceil\log (q+1)\rceil \ge \log (p+1)+\log (q+1) \ge 2 \log(p+q)$, the last inequality holds for all $p\ge 1$ and $q \ge 1$. Therefore, the number of bits required to represent $x^*$ is
\[
\begin{split}
sp(x^*)&=\lceil\log(A_1 p+B_1 q+1)\rceil+\lceil\log(C p+D q+1)\rceil\leq 2 \lceil\log(E(p+q)+1)\rceil\leq\\
	&\leq 2 \log E+2\log(p+q)+2\leq 2+2c \log n+sp(x_d)\,,
\end{split}
\]
where $E=\max\{A_1,B_1,C,D\}$. Analogously for $y^*$, $sp(y^*)\leq2+2c \log n+sp(x_d)$. Therefore, at every step, the bit complexity of the coordinates grows no more than by an additive value $2+2c \log n$. After $n$ steps, the bit-complexity of the regions' vertices is $O(n\log n)$.

\begin{theorem} \label {thm:bit-up-regions}
Representing the regions $\mathcal{R}$ as polygons with rational coordinates requires $O(n^2 \log n)$ bits.
\end{theorem}
\frank{should this really be a theorem? A Lemma seems more appropriate}

\begin {corollary} \label {cor:bit-up-paths}
If there exists a solution with $k$ links, there also exists one in which the coordinates of the bends use at most $O(k \log n)$ bits.
\end {corollary}


\begin{theorem} \label {2d-np}
\mlp{a}2 is in NP.
\end{theorem}

\begin {proof}
  We need to show that a candidate solution can be verified in polynomial time.
  A potential solution needs at most $n$ links.
  By Corollary~\ref{cor:bit-up-paths}, we only need to verify candidate solutions that consist
  of bends with $O(n \log n)$-bit coordinates.
  Given such a candidate, we need to verify pairwise visibility between at most $n$ pairs of
  points with $O(n \log n)$-bit coordinates, which can be done in polynomial time.
\end {proof}


\section{Computational Complexity in \texorpdfstring{$\R^2$}{R2}}
\label{sec:hardness2d}

In this section we show that \mlp12 is NP-hard. To this end, we first provide a
blueprint for our reduction in Section~\ref{sub:blueprint}. 
In Section~\ref{sub:mlp12} we then show how to
``instantiate'' this blueprint for \mlp12 in a polygon with holes.

\subsection {A Blueprint for Hardness Reductions}
\label{sub:blueprint}

We reduce from the 2-Partition problem: Given a set of integers
$A=\{a_1,\dots,a_m\}$, find a subset $S\subseteq A$ whose sum is equal to half
the sum of all numbers. The main idea behind all the hardness reductions is as
follows. Consider a 2D construction in Fig.~\ref{fig:hard}~(left). Let point
$s$ have coordinates $(0,0)$, and $t$ (not in the figure) have coordinates
$(\sum a_i/2,4m-2)$. For now, in this construction, we will consider only paths
from $s$ to $t$ that are allowed to bend on horizontal lines with even
$y$-coordinates. Moreover, we will count an intersection with each such
horizontal line as a bend. We will place fences along the lines with odd
$y$-coordinates in such a way that an $s$-$t$ path with $2m-1$ links exists
(that bends only on horizontal lines with even $y$-coordinates) if and only if
there is a solution to the 2-Partition instance.

Call the set of horizontal lines $\ell_0:y=0$, $\ell_i:y=4i-2$ for $1\leq i \leq m$ \emph{important} (dashed lines in Fig.~\ref{fig:hard}), and the set of horizontal lines $\ell'_i:y=4i-4$ for $2\leq i \leq m$ \emph{intermediate} (dash-dotted lines in Fig.~\ref{fig:hard}). Each important line $\ell_i$ will ``encode'' the running sums of all subsets of the first $i$ integers $A_i=\{a_1,\dots,a_i\}$. That is, the set of points on $\ell_i$ that are reachable from $s$ with $2i-1$ links will have coordinates $(\sum_{a_j\in S_i}a_j, 4i-2)$ for all possible subsets $S_i\subseteq A_i$.

Call the set of horizontal lines $f_1:y=1$, $f_i:y=4i-5$ for $2\leq i \leq m$
\emph{multiplying}, and the set of horizontal lines $f'_i:y=4i-3$ for $2\leq i
\leq m$ \emph{reversing}. Each multiplying line $f_i$ contains a fence with two
$0$-width slits that we call $0$-slit and $a_i$-slit. The $0$-slit with
$x$-coordinate $0$ corresponds to not including integer $a_i$ into subset
$S_i$, and the $a_i$-slit with $x$-coordinate $\sum_1^{i}{a_j}-a_i/2$
corresponds to including $a_i$ into $S_i$. Each reversing line $f'_i$ contains
a fence with two $0$-width slits (reversing $0$-slit and reversing $a_i$-slit)
with $x$-coordinates $0$ and $\sum_1^i{a_j}$ that ``put in place'' the next
bends of potential min-link paths, i.e., into points on $\ell_i$ with
$x$-coordinates equal to running sums of $S_i$.  We add a vertical fence of
length $1$ between lines $\ell'_i$ and $f'_i$ at $x$-coordinate
$\sum_1^i{a_j}/2$ to prevent the min-link paths that went through the
multiplying $0$-slit from going through the reversing $a_i$-slit, and vice versa.

As an example, consider (important) line $\ell_2$ in Fig.~\ref{fig:hard}. The four points on $\ell_2$ that are reachable from $s$ with $3$ links have $x$-coordinates $\{0, a_1, a_2, a_1+a_2\}$. The points on line $\ell'_3$ that are reachable from $s$ with a path (with $4$ links) that goes through the $0$-slit on line $f_3$ have $x$-coordinates $\{0,-a_1,-a_2,-(a_1+a_2)\}$, and the points on $\ell'_3$ that are reachable from $s$ through the $a_3$-slit have $x$-coordinates $\{a_1\!+\!a_2\!+\!a_3,2a_1\!+\!a_2\!+\!a_3,a_1\!+\!2a_2\!+\!a_3,2a_1\!+\!2a_2+a_3\}$. The reversing $0$-slit on line $f'_3$ places the first four points on $\ell_3$ into $x$-coordinates $\{0, a_1, a_2, a_1+a_2\}$, and the reversing $a_3$-slit places the second four points on $\ell_3$ into $x$-coordinates $\{a_3, a_1+a_3, a_2+a_3, a_1+a_2+a_3\}$.

In general, consider some point $p$ on line $\ell_{i-1}$ that is reachable from
$s$ with $2i-3$ links. The two points on $\ell'_i$ that can be reached from $p$
with one link have $x$-coordinates $-p_x$ and $2\sum_1^{i}{a_j}-a_i-p_x$, where
$p_x$ is the $x$-coordinate of $p$. Consequently, the two points on $\ell_i$
that can be reached from $p$ with two links have $x$-coordinates $p_x$ and
$p_x+a_i$. Therefore, for every line $\ell_i$, the set of points on it that are
reachable from $s$ with a min-link path have $x$-coordinates equal to
$\sum_{a_j\in S_i}a_j$ for all possible subsets $S_i\subseteq A_i$. Consider
line $\ell_m$ and the destination point $t$ on it. There exists a $s$-$t$ path
with $2m-1$ links if and only if the $x$-coordinate of $t$ is equal to
$\sum_{a_j\in S}a_j$ for some $S\subseteq A$. The complexity of the
construction is polynomial in the size of the 2-Partition instance. Therefore,
finding a min-link path from $s$ to $t$ in our 2D construction is $NP$-hard.


\begin{figure}[t]
\centering
\includegraphics[page=2]{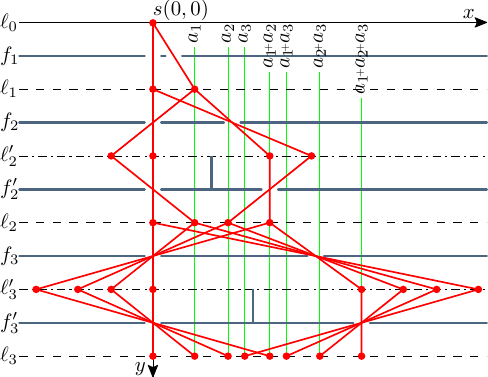}
\caption{The first few lines of a 2D construction depicting the general idea behind the hardness proofs: important lines $\ell_0$--$\ell_3$, intermediate lines $\ell'_1$--$\ell'_3$, multiplying lines $f_1$--$f_3$, and reversing lines $f'_1$--$f'_3$. The slits in the fences on multiplying and reversing lines are placed in such a way that the locations on $\ell_i$ that are reachable from $s$ with $2i-1$ links correspond to sums formed by all possible subsets of $\{a_1,\dots,a_i\}$.
}
\label{fig:hard}
\end{figure}

\subsection {Hardness of \texorpdfstring{\mlp12}{MinLinkPath1,2}}
\label{sub:mlp12}

We can turn our construction from Section~\ref{sub:blueprint} into a ``zigzag''
polygon (Fig.~\ref{fig:diffuse}); the fences are turned into obstacles within
the corresponding corridors, and slits remain slits---the only free space
through which it is possible to go with one link between the polygon edges that
correspond to consecutive lines $\ell'_i$ and $\ell_i$ (or $\ell_{i-1}$ and
$\ell'_i$). This retains the crucial property of 2D construction: locations
reachable with fewest links on the edges of the polygon correspond to sums of
numbers in the subsets of $A$. We conclude:

\begin{theorem}
  \label{thm:mlp12_2d}
  \mlp12 in a 2D polygonal domain with holes is NP-hard.
\end{theorem}

\begin{figure}[t]
\centering
\includegraphics{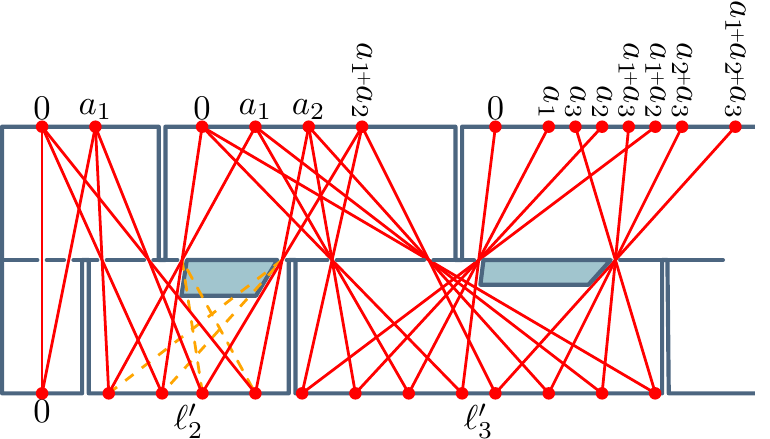}
\caption{There exists an \s-\t diffuse reflection path with $2m-1$ links iff 2-Partition instance is feasible.}
\label{fig:diffuse}
\end{figure}

Overall our reduction bears resemblance to the classical \e{path encoding} scheme \cite{3dhard} used to prove hardness of 3D shortest path and other path planning problems, as we also repeatedly double the number of path homotopy types; however, since we reduce from 2-Partition (and not from 3SAT, as is common with path encoding), our proof(s) are much less involved than a typical path-encoding one.

\paragraph{No FPTAS.} Obviously, problems with a discrete
answer (in which a second-best solution is separated by at least 1 from the
optimum) have no FPTAS. \frank{if this is so obvious, we probably should not
  spend 15 lines on it...}For example, in the reduction in
Theorem~\ref{thm:mlp12_2d}, if the instance of 2-Partition is feasible, the
optimal path has $2m-1$ links; otherwise it has $2m$ links. Suppose there
exists an algorithm, which, for any $\eps>0$ finds a $(1+\eps)$-approximate
solution in time polynomial in $1/\eps$. Take $\eps=\frac{1}{2m-1}$; note that
1/\eps is polynomial, and hence the FPTAS with this \eps will complete in
polynomial time. For an infeasible instance of 2-Partition the FPTAS would
output a path with at least $2m$ links, while for a feasible instance it will
output a path with at most $(1+\eps)(2m-1)=2m-1/2$ links. There is only one
such length possible; a path with exactly $\opt = 2m-1$ links. Hence, the FPTAS
would be able to differentiate, in polynomial time, between feasible and
infeasible instances of 2-Partition.

\paragraph{No additive approximation. } We can slightly amplify the hardness
results, showing that for any constant $K$ it is not possible to find an
additive-$K$ approximation for our problems: Concatenate $K$ instances of the
construction from the hardness proof, aligning $s$ in the instance $k+1$ with
$t$ from the instance $k$. Then there is a path with $K(2m-1)$ links through
the combined instance if the 2-Partition is feasible; otherwise $K(2m-1)+K-1$
links are necessary, Thus an algorithm, able to differentiate between instances
in which the solution has $K(2m-1)$ links and those with $K(2m-1)+K-1$ links in
$\textrm{poly}(mK)=\textrm{poly}(m)$ time, would also be able to solve
2-Partition in the same time.

\section{Algorithmic Results in \texorpdfstring{$\R^2$}{R2}}
\label{sec:apxs2d}


\subsection {Constant-factor Approximation}

\mlp22 can be solved exactly \cite{mrw}.
For \mlp12, \cite{diffuse} gives a $3$-approximation.

\subsection{PTAS}\label{sec:ptas2d}

We describe a $(1+\eps)$-approximation scheme for \mlp12, based on building a graph of edges of $D$ that are \e{$k$-link weakly visible}.

Consider the set $F$ of all edges of $D$ (that is, $\bigcup F = \ske1$).
To avoid confusion between edges of $D$ and edges of the graph we will build, we will call elements of $F$ \emph {features} (this will also allow us to extend the ideas to higher dimensions later).
Two features $f,f' \in F$ are \e{weakly visible} if there exist mutually visible points $p\in f,p'\in f'$;
more generally, we say $f,f'$ are \e{$k$-link weakly visible} if there exists a $k$-link path from $p$ to $p'$ (with the links restricted to $\ske1$).

For any constant $k\ge1$, we construct a graph $G^k = (F, E_k)$, where $E_k$ is the set of pairs of $k$-link weakly visible features.
Let $\pi^k=\{f_0,f_1,\dots,f_\ell\}$, with $f_0 \ni \s$ and $f_\ell \ni \t$ be a shortest path in $G$ from the feature containing $\s$ to the feature containing $\t$; $\ell$ is the number of links of $\pi$.
We describe how to transform $\pi^k$ into a solution to the \mlp12 problem.
Embed edges of $\pi$ into $D$ as $k$-link paths. This does not necessarily connect \s to \t since it could be that, inside a feature $f_i$, the endpoint of the edge $f_{i-1}f_i$ does not coincide with endpoint of the edge $f_if_{i+1}$; to create a connected path, we observe that the two endpoints can always be connected by two extra links via some feature that is mutually visible from both points (or a single extra link within $f_i$ if we allow links to coincide within the boundary of $D$).

\begin{lemma}The number of links in $\pi^k_*$ is at most $(1+1/k)\opt$.\end{lemma}
\begin{proof}Split \opt into pieces of $k$ links each (the last piece may have fewer than $k$ links); the algorithm will find $k$-link subpaths between endpoints of the pieces. In details, suppose that $\opt=mk+r$ where $m,r$ are the quotient and the remainder from division of $\opt$ by $k$; let $s=v_0,v_1,\dots,v_\opt=t$ be the vertices (bends) of \opt, and let $f_i$ be the feature to which the $ik$-th bend $v_{ik}$ belongs. Since the link distance between $v_{(i-1)k}$ and $v_{ik}$ is $k$, our algorithm will find a $k$-link subpath from $f_{i-1}$ to $f_i$, as well as an $r$-link subpath from $f_m$ to \t. The total number of links in the approximate path is thus at most $mk+m+r\le(1+1/k)(mk+r)=(1+1/k)\opt$ (if $r=0$, our algorithm will find path with at most $mk+m-1<(1+1/k)mk=(1+1/k)\opt$ links; if $r>0$, our algorithm will find path with at most $mk+r+m\le(1+1/k)(mk+r)=(1+1/k)\opt$ links).
\end{proof}

We now argue that the weak $k$-link visibility between features can be determined in polynomial time using the staged illumination: starting from each feature $f$, find the set $W(f)$ of points on other features weakly visible from $f$, then find the set weakly visible from $W^2(f)=W(W(f))$, repeat $k$ times to obtain the set $W^k(f)$ reachable from $f$ with $k$ links; feature $f'$ can be reached from $f$ in $k$ links iff $W^k(f)\cap f'\ne\emptyset$. For constant $k$, building $W^k(f)$ takes time polynomial in $n$, although possibly exponential in $k$ (in fact, for diffuse reflection explicit bounds on the complexity of $W^k(f)$ were obtained \cite{aronov1,aronovk,n9}). This can be seen by induction: Partition the set $W^{i-1}(f)$ into the polynomial number of constant-complexity pieces. For each piece $p$, each element $e$ of the boundary of the domain and each feature $f'$ compute the part of $f'$ shadowed by $e$ from the light sources on $p$---this can be done in constant time analogously to determining weak visibility between two features above (by considering the part of $p\times f'$ carved out by the occluder $e$). The part of $f'$ weakly seen from $W^{i-1}(f)$ is the union, over all parts $p$, of the complements of the sets occluded by all elements $e$; since there is a polynomial number of parts, elements and features, it follows that $W^i(f)$ can be constructed in polynomial time.

\begin{theorem}
\label{thm:ptas2d}
For a constant $k$ the path $\pi^k_*$, having at most $(1+1/k)\opt$ links, can be constructed in polynomial time.
\end{theorem}



\section {Algebraic Complexity in \texorpdfstring{$\R^3$}{R3}}
\label{sec:bit3d}

\subsection{Lower bound on the Bit complexity}

\maarten {Here we say that the lower bound from $\R^2$ obviously extends, and that we get curves on the boundary of the reachable regions so the upper bound does not extend.}

\subsection{Upper bound on the algebraic complexity}

\paragraph{Order of the boundary curves.} 
Assume the representations of the coordinates of any vertex of $D$ and $s$ require at most $c_0\log n$ bits for
some constant $c_0$. Analogous to Section~\ref{sec:bit2d}, we define a sequence
of regions $\mathcal{R}=\{R_1,R_2,R_3,\dots\}$, where $R_1$ is the set of all
points in $D$ that see $s$, and $R_{i}$ is the region of points in $D$ that see
some point in $R_{i-1}$ for $i\geq 2$, i.e., region $R_{i}$ consists of all the
points of $D$ that are illuminated by region $R_{i-1}$. Note, that $R_i$ is a
union of subsets of faces of $D$. Therefore, when we will speak of the
boundaries (in the plural form of the word) of $R_i$, that we denote as
$\partial R_i$, we will mean the illuminated sub-intervals of edges of $D$ as
well as the frontier curves interior to the faces of $D$.

\begin{figure}[t]
\centering
\includegraphics[width=0.4\textwidth]{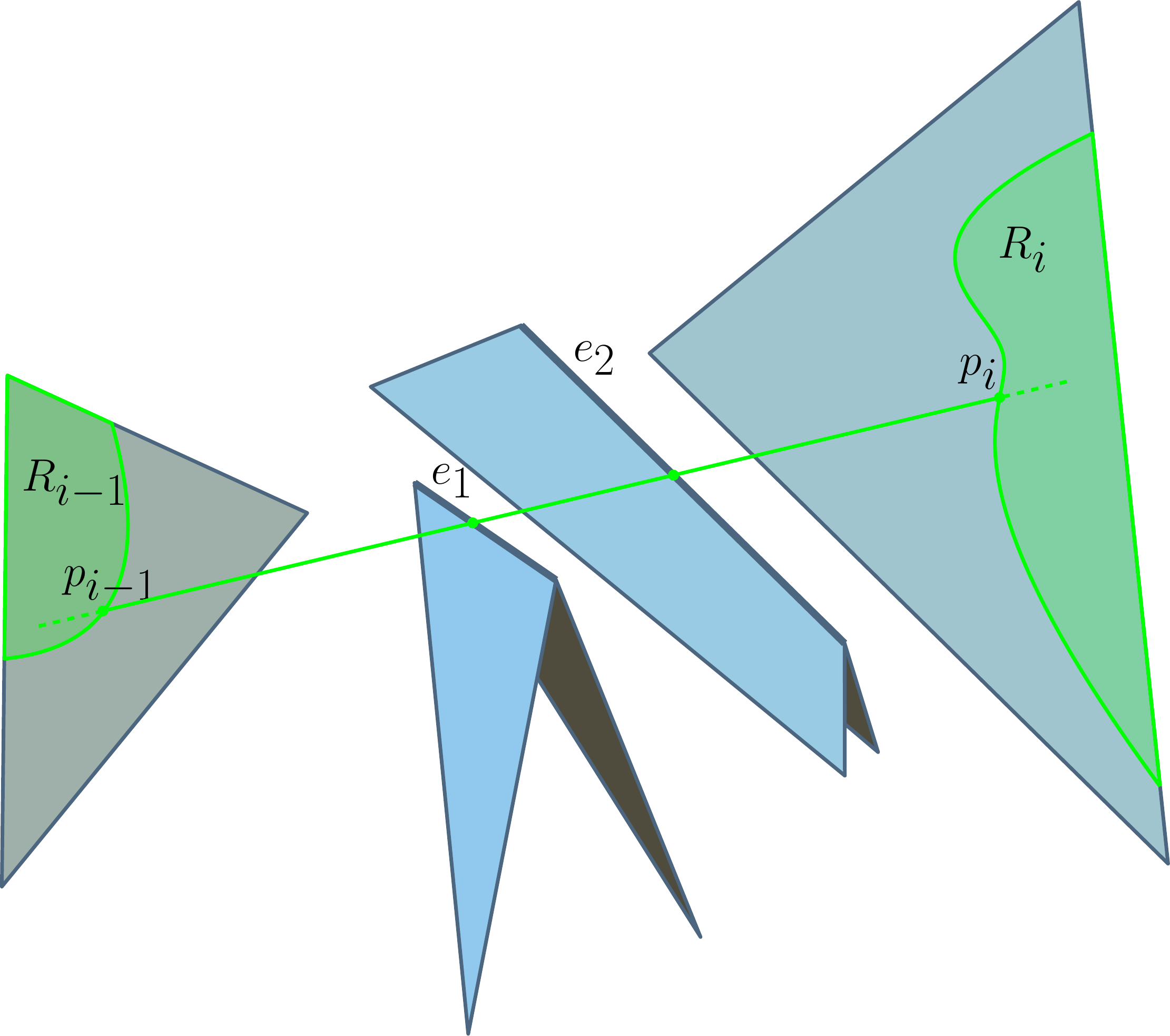}
\caption{The order of the curves on the boundaries of $R_i$ grows with $i$.}
\label{fig:3d-algeb-comp}
\end{figure}

Unlike in 2D, the boundaries of $R_i$ interior to the faces of $D$ do not necessarily consist of straight-line segments. Observe, that a union of all lines intersecting three given lines in 3D forms a hyperboloid, and therefore, illuminating a straight-line segment on the boundaries of $R_{i-1}$ leads to the corresponding part of $\partial R_{i}$ to be an intersection of a hyperboloid and a plane, i.e., a hyperbola. Moreover, consider some point $p_{i-1}\in\partial R_{i-1}$ interior to some face $f_{i-1}$ of $D$, and two edges $e_1$ and $e_2$ of the domain $D$ which $p_{i-1}$ sees partially and which will cast a shadow on some face $f_i$ of $D$ (refer to Fig.~\ref{fig:3d-algeb-comp}). Then we can express the coordinates of $p_{i}$ as:
\begin{equation}\label{eq:piofpi-1}
\begin{pmatrix}
    x_{i} \\
    y_{i} \\
    z_{i}
\end{pmatrix}=
\begin{pmatrix}
\dfrac{A_1 x_{i-1}^2+B_1 y_{i-1}^2+C_1 x_{i-1} y_{i-1} +D_1 x_{i-1} + E_1 y_{i-1} +F_1}{A x_{i-1}^2+B y_{i-1}^2+C x_{i-1} y_{i-1} +D x_{i-1} + E y_{i-1} +F}\\[12pt]
\dfrac{A_2 x_{i-1}^2+B_2 y_{i-1}^2+C_2 x_{i-1} y_{i-1} +D_2 x_{i-1} + E_2 y_{i-1} +F_2}{A x_{i-1}^2+B y_{i-1}^2+C x_{i-1} y_{i-1} +D x_{i-1} + E y_{i-1} +F}\\[12pt]
U x_{i} +V y_{i}+W
\end{pmatrix}
\,,
\end{equation}
for some constants $A_1,A_2,A,B_1,\dots,U,V,W$ that depend on the parameters of $f_{i-1}$, $f_i$, $e_1$, $e_2$. Denote a polynomial of degree $d$ as $\mathsf{poly}^d(\cdot)$, then we can rewrite the $x$- and the $y$-coordinates of $p_i$ as
\[
\begin{pmatrix}
    x_{i} \\
    y_{i}
\end{pmatrix}=
\begin{pmatrix}
\dfrac{\mathsf{poly}_{x,i-1}^2(x_{i-1},y_{i-1})}{\mathsf{poly}_{i-1}^2(x_{i-1},y_{i-1})} \\[12pt]
\dfrac{\mathsf{poly}_{y,i-1}^2(x_{i-1},y_{i-1})}{\mathsf{poly}_{i-1}^2(x_{i-1},y_{i-1})}
\end{pmatrix}=
\begin{pmatrix}
 \dfrac{\mathsf{poly}_{x,i-2}^4(x_{i-2},y_{i-2})}{\mathsf{poly}_{i-2}^4(x_{i-2},y_{i-2})} \\[12pt]
 \dfrac{\mathsf{poly}_{y,i-2}^4(x_{i-2},y_{i-2})}{\mathsf{poly}_{i-2}^4(x_{i-2},y_{i-2})}
\end{pmatrix}=
\begin{pmatrix}
 \dfrac{\mathsf{poly}_{x,0}^{2i}(x_0,y_0)}{\mathsf{poly}_0^{2i}(x_0,y_0)}\\[12pt]
 \dfrac{\mathsf{poly}_{y,0}^{2i}(x_0,y_0)}{\mathsf{poly}_0^{2i}(x_0,y_0)}
\end{pmatrix}\,,
\]
\frank{what do the subscripts to poly mean?}
where point $p_0(x_0,y_0,z_0)$ lies on some straight-line segment of $\partial D$, and we use different subscripts of the polynomials to distinguish between different expressions. Notice that the denominators of the $x_i$ and $y_i$ expressed as functions of $x_j$ and $y_j$ (for all $j<i$) are always the same. If we slide $p_0$ along the line segment, and express its coordinates in terms of a parameter $t$, we get
\[
x_{i}= \frac{\mathsf{poly}_{x}^{2i}(t)}{\mathsf{poly}^{2i}(t)}\,, \qquad y_{i}= \frac{\mathsf{poly}_{y}^{2i}(t)}{\mathsf{poly}^{2i}(t)}\,,\qquad z_{i}= \mathsf{poly}^1(x_i,y_i)\,.
\]
Thus, the curve, that point $p_i$ traces on $f_i$ is an intersection of a plane in 3D (face $f_i$) and two surfaces of order $2i+1$ in 4D space (with coordinates $x$, $y$, $z$, and $t$). Therefore, the order of that curve is not greater than $2i+1$. In fact, as we have mentioned above, for $i=1$, the curve that $p_1$ traces on face $f_1$ is a hyperbola, with order $2$, and not $2i+1=3$. The fact that the denominators of the expressions of $x_1$ and $y_1$ are the same allow to reduce the order of the expressions in the following way:
\begin{equation}\label{eq:hyperbola}
\begin{split}
x_{1}&= \frac{\mathsf{poly}_{x}^{2}(t)}{\mathsf{poly}^{2}(t)}=x'_1+\frac{\mathsf{poly}_{x'}^{1}(t)}{\mathsf{poly}^{2}(t)}\,,\\
y_{1}&= \frac{\mathsf{poly}_{y}^{2}(t)}{\mathsf{poly}^{2}(t)}=y'_1+\frac{\mathsf{poly}_{y'}^{1}(t)}{\mathsf{poly}^{2}(t)}\,,
\end{split}
\end{equation}
Therefore,
\[
\frac{x_1-x'_1}{y_1-y'_1}=\frac{\mathsf{poly}_{x'}^{1}(t)}{\mathsf{poly}_{y'}^{1}(t)}\,,
\]
and then
\[
t=poly^1(x_1,y_1)\,.
\]
Substituting this expression into Equations~\ref{eq:hyperbola} we get, that the actual order of the curve traced by $p_1$ is $2$. For larger $i$, denominators of the expressions of $x_i$ and $y_i$ are also equal, however the explicit formula for the curve traced by $p_i$ cannot be derived in a similar way. We summarize our findings:
\begin{theorem}
The boundaries of region $R_i$ are curves of order at most $2i+1$ for $i\geq 2$, and at most $2$ for $i=1$.
\end{theorem}

The fact that the order of the curves on the boundaries of $R_i$ grows linearly may give hope that the bit complexity of representation of $R_i$ can be bounded from above similarly to Section~\ref{sec:bit2d-up}. However, following similar calculations we will get that the space required to store the coordinates of $p_i$ grows exponentially with $i$.

Parameters $A_1,A_2,A,B_1,\dots,W$ of Equation~\ref{eq:piofpi-1} have bit complexity not greater than $c \log n$ for some constant $c$ \irina{We probably can leave out here that $c=5 c_0+\log88$ because of some crazy formulas that will take a few pages in the appendix?}. Let $x_{i-1}$ be represented as a rational number $\nicefrac{p_x}{q_x}$, and $y_{i-1}$ be represented as a rational number $\nicefrac{p_y}{q_y}$, where $p_x$ and $q_x$, and $p_y$ and $q_y$ are two pairs of mutually prime integers. Then the number of bits required to represent $x_{i-1}$ is $sp(x_{i-1})\ge \max\{\log p_x,\log q_x\}$. Therefore, the number of bits required to represent $x_i$
\[
\begin{split}
sp(x_i)\le & \log(A_1 p_x^2 q_y^2+B_1 q_x^2 p_y^2+C_1 p_x q_x p_y q_y + D_1 p_x q_x q_y^2+E_1 q_x^2 p_y q_y + F_1 q_x^2 q_y^2)+\\
	& \log(A p_x^2 q_y^2+B q_x^2 p_y^2+C p_x q_x p_y q_y + D p_x q_x q_y^2+E q_x^2 p_y q_y + F q_x^2 q_y^2)\leq\\
\le	& 2\log(6 M r^4)=2\log 6+2\log M+8\log r \le 6+2 c\log n + 8 \max\{sp(x_{i-1}),sp(y_{i-1})\}\,,
\end{split}
\]
where $M=\max\{A_1,B_1,\dots,E,F\}$ and $r=\max\{p_x,q_x,p_y,q_y\}$. Solving the above recurrence we get $sp(x_i)\leq 9^i$, which implies exponential upper bound of the space required to store $x_i$.

\begin{lemma}\label{lem:bit3d-up}
The coordinates of a vertex of $R_i$ can be stored in $O(9^i)$ space.
\end{lemma}


\section {Computational Complexity in \texorpdfstring{$\R^3$}{R3}}
\label{sec:hardness3d}

We will show now how to use our blueprint from Section~\ref{sub:blueprint} to
build a terrain for the \mlp12 problem such that a path from $s$ to $t$ with
$2n-1$ links will exist if and only if there exists a subset $S\subseteq A$
whose sum is equal to half the sum of all integers $A=\{a_1,\dots,a_m\}$. Take
the 2D construction and bend it along all the lines $\ell_i$ and $\ell'_i$,
except $\ell_0$ and $\ell_m$ (refer to Fig.~\ref{fig:hard3D}). Let the
angles between consecutive faces be $\pi-\delta$ for some small angle
$\delta<\pi/4m$ (so that the sum of bends between the first face (between the
lines $\ell_0$ and $\ell_1$) and the last face (between the lines $\ell'_m$ and
$\ell_m$) is less than $\pi$). On each face build a fence of height
$\tan(\delta/4)$ according to the 2D construction. The height of the fences
is small enough so that no two points on consecutive fences see each
other. Therefore, for two points $s$ and $t$ placed on $\ell_0$
and $\ell_m$ as described above, an $s$-$t$ path with $2m-1$ links must bend
only on $\ell_i$ and $\ell'_i$ and pass in the slits in the fences. Finding a
min-link path on such a terrain is equivalent to finding a min-link path (with
bends restricted to $\ell_i$ and $\ell'_i$) in the 2D
construction. Therefore,

\begin{figure}[t]
\centering
\includegraphics[width=0.48\textwidth,page=1]{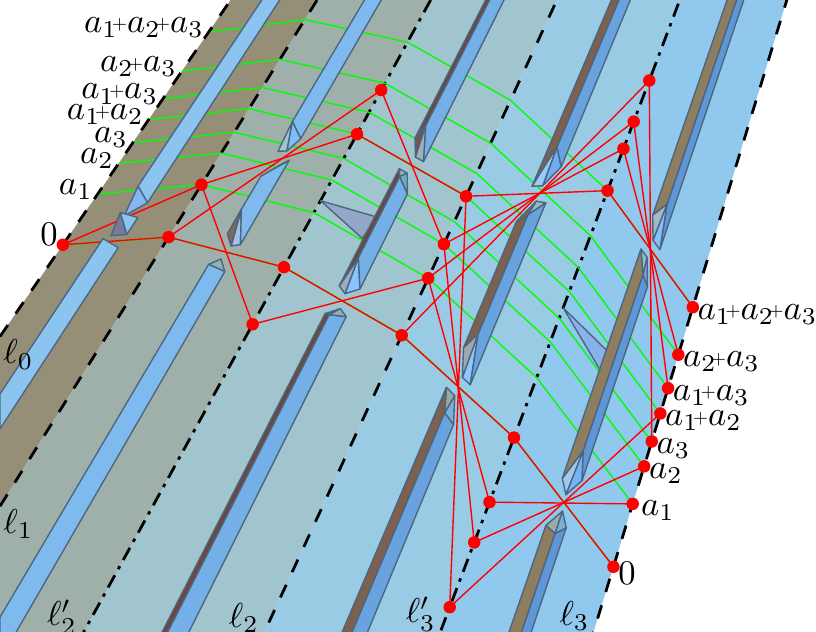}
\caption{
Right: The terrain obtained by bending the 2D construction along the important and intermediate lines. The height of the fences is low enough that no two points on consecutive fences can see each other.
}
\label{fig:hard3D}
\end{figure}

\begin{theorem}
  \label{thm:hard}
  \mlp12 on a terrain is NP-hard.
\end{theorem}

\begin{remark}
  Instead of $0$-width slits, we could use slits of positive width $ w=o (
  \frac{1}{4m} )$; since the width of the light beam grows by $2w$ between two
  consecutive creases, on the last crease the maximum shift of the path due to
  the positive slits width will be at most $ (2m-1)\times 2w<1$.  \maarten {Do
    we need the shift to be $< \frac12$ to make sure two points cannot shift to
    the same place?}
\end{remark}

Observe that bending in the interior of a face cannot reduce the link distance
between $s$ and $t$. Hence, our reduction also shows that \mlp22 is
NP-hard. Furthermore, lifting the links from the terrain surface into $\R^3$
also does not reduce link distance; we can make sure that the fences are low in
height, so that fences situated on different faces of the creased rectangle do
not see each other. Therefore, jumping onto the fences is useless. Hence,
\mlp13 and \mlp23 are also NP-hard.

\paragraph{\mlp{a}{b} in general polyhedra} Since a terrain is a special case
of a 3D polyhedra, it follows that \mlp12, \mlp22, \mlp13, and \mlp23 are also
NP-hard for an arbitrary polyhedral domain in $\R^3$. Our construction does not
immediately imply that \mlp33 is NP-hard. However, we can put a copy of the
terrain slightly above the original terrain (so that the only free space is the
thin layer between the terrains). When this layer is thin enough, the ability
to take off from the terrain, and bend in the free space, does not help in
decreasing the link distance from \s to \t. Thus, \mlp33 is also NP-hard.

\begin{corollary}
  \mlp{a}{b}, with $a \geq 1$ and $b \geq 2$, in a 3D domain $D$ is NP-hard. This holds even
  if $D$ is just a terrain.
\end{corollary}

\section{Algorithmic Results in \texorpdfstring{$\R^3$}{R3}}
\label{sec:apxs3d}


\subsection {Constant-factor Approximation}
\label{sec:const3d}

Our approximations refine and extend the 2-approximation for minimum-link paths in higher dimensions suggested in Chapter~26.5 (section Other Metrics) of the handbook \cite{handbook04} (see also Ch.~6 in \cite{piatko}); since the suggestion is only one sentence long, we fully quote it here:\begin{quote}Link distance in a polyhedral domain in $\mathbb R^d$ can be approximated (within factor 2) in polynomial time by searching a weak visibility graph whose nodes correspond to simplices in a simplicial decomposition of the domain.\end{quote}

Indeed, consider $\ske{a}$, the set of all points where the path is allowed to bend, and
decompose $\ske{a}$ into a set $F$ of small-complexity convex pieces, and call each piece a \e{feature}. Similar to Section~\ref {sec:ptas2d}, we say two features $f$ and $f'$ are \e{weakly visible} if there exist mutually visible points $p\in f$ and $p'\in f'$; more generally, the \e{weak visibility} region $W(f)$ is the set of points that see at least one point of $f$, so $f'$ is weakly visible from $f$ iff $f'\cap W(f)\ne\emptyset$ (in terms of \e{illumination} $W(f)$ is the set of points that get illuminated when a light source is put at every point of $f$).  See Fig.~\ref{fig:param_space} for an illustration.
\begin{figure}
\centering
  \includegraphics{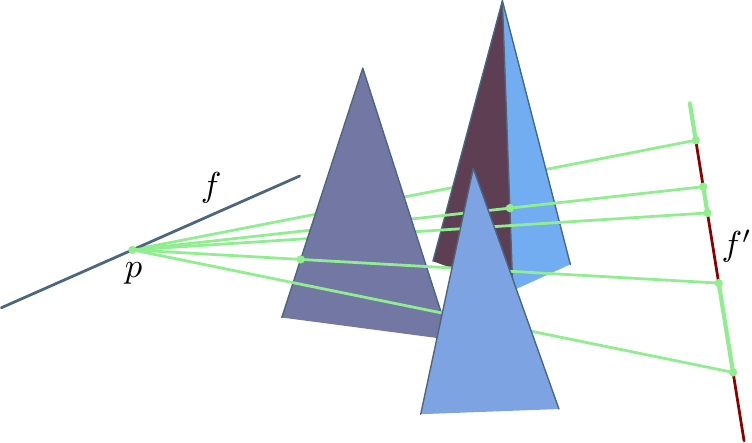}
  \caption{The weak visibility $W(f)$ restricted to edge $f'$ is the union of
    all visible intervals (green) over all points $p \in f$. If this region is
    non-empty, $f$ and $f'$ are weakly visible.}
  \label{fig:param_space}
\end{figure}

Weak visibility between two features $f$ and $f'$ can be determined straightforwardly by building the set of pairs of points $(p,p')$ in the parameter space $f\times f'$ occluded by (each element of) the obstacles. To be precise, $f\times f'$ is a subset of $\R^{2a}$.
Now, consider $\ske{d-1}$, which we also decompose into a set of constant-complexity \e{elements}. Each element $e$ defines the set $B(e)=\{(p,p')\in f\times f': pp'\cap e\ne\emptyset\}$ of pairs of points that it blocks; since $e$ has constant complexity, the boundary of $B(e)$ consists of a constant number of curved surfaces, each described by a low degree polynomial. Since there are $O(n)$ elements, the union (and, in fact, the full arrangement) of the sets $B(e)$ for all $e$ can be built in $O(n^{4a-3+\eps})$ time, for an arbitrarily small $\eps > 0$, or $O(n^2)$ time in case $a = 1$~\cite{agarwal2000arrangements}. We define the \e{visibility map} $M(f,f')\subseteq f\times f'$ to be the complement of the union of the blocking sets, i.e., the map is the set of mutually visible pairs of points from $f\times f'$. We have:
\begin{lemma}
\label{lem:map}
$M(f,f')$ can be built in $O(n^{\max (2,4a-3+\eps)})$ time, for an arbitrarily small $\eps > 0$.
\end{lemma}

The features $f$ and $f'$ weakly see each other iff $M(f,f')$ is not empty. Let $G$ be the graph on features whose edges connect weakly visible features; \s and \t are added as vertices of $G$, connected to features (weakly) seen from them. Let $\pi=\{f_0,f_1,\dots,f_\ell\}$, with $f_0 = \s$ and $f_\ell = \t$ be a shortest \s-\t path in $G$; $\ell$ is the length of $\pi$. Embed edges of $\pi$ into the geometric domain, putting endpoints of the edges arbitrarily into the corresponding features. This does not necessarily connect \s to \t since it could be that, inside a feature $f_i$, the endpoint of the edge $f_{i-1}f_i$ does not coincide with endpoint of the edge $f_if_{i+1}$; to create a connected path, connect the two endpoints by an extra link within $f_i$ (this is possible since the features are convex).

Bounding the approximation ratio of the above algorithm is straightforward: Let \opt denote a min-link \s-\t path and, abusing notation, also the number of links in it. Consider the features to which consecutive bends of \opt belong; the features are weakly visible and hence are adjacent in $G$. Thus $\ell\le\opt$. Adding the extra links inside the features adds at most $\ell-1$ links. Hence the total number of links in the produced path is at most $2\ell-1<2\opt$.

Since $G$ has $O(n)$ nodes and $O(n^2)$ edges, Dijkstra's algorithm will find
the shortest path in it in $O(n^2)$ time.

\begin{theorem}
\label{thm:2}
(cf.\ \cite[Ch.~27.5]{handbook04}.) A 2-approximation to \mlp{a}{b} can be found in \linebreak $O(n^{2+\max(2,4a-3+\eps)})$ time, where $\eps > 0$ is an arbitrarily small constant.
\end{theorem}

Interestingly, the running time in Theorem~\ref{thm:2} depends only on $a$, and not on $b$ or $d$, the dimension of $D$ (of course, $a\le d$, so the runtime is bounded by $O(n^{2+\max(2,4d-3+\eps)})$ as well).

\subsection{PTAS}

To get a $(1+1/k)$-approximation algorithm for any constant $k\ge1$, we expand the above handbook idea by searching for shortest \s-\t path $\pi^k$ in the graph $G^k$ whose edges connect features that are \e{$k$-link weakly visible}. 
Similarly to Section~\ref{sec:ptas2d}, we obtain the following.

\begin{theorem}
\label{thm:ptas3d}
For a constant $k$ the path $\pi^k_*$, having at most $(1+1/k)\opt$ links, can be constructed in polynomial time.
\end{theorem}
\begin{proof}The approximation factor follows from the same argument as in Section~\ref{sec:ptas2d}. To show the polynomial running time, we argue that the weak $k$-link visibility between features can be determined in polynomial time using the staged illumination: starting from each feature $f$, find the set $W(f)$ of points on other features weakly visible from $f$, then find the set weakly visible from $W^2(f)=W(W(f))$, repeat $k$ times to obtain the set $W^k(f)$ reachable from $f$ with $k$ links; feature $f'$ can be reached from $f$ in $k$ links iff $W^k(f)\cap f'\ne\emptyset$. For constant $k$, building $W^k(f)$ takes time polynomial in $n$, although possibly exponential in $k$ (in fact, for diffuse reflection explicit bounds on the complexity of $W^k(f)$ were obtained \cite{aronov1,aronovk,n9}). This can be seen by induction: Partition the set $W^{i-1}(f)$ into the polynomial number of constant-complexity pieces. For each piece $p$, each element $e$ of the boundary of the domain and each feature $f'$ compute the part of $f'$ shadowed by $e$ from the light sources on $p$---this can be done in constant time analogously to determining weak visibility between two features above (by considering the part of $p\times f'$ carved out by the occluder $e$). The part of $f'$ weakly seen from $W^{i-1}(f)$ is the union, over all parts $p$, of the complements of the sets occluded by all elements $e$; since there is a polynomial number of parts, elements and features, it follows that $W^i(f)$ can be constructed in polynomial time.
\end{proof}

\subsection{The global visibility map of a terrain}
\label{sec:vis}

Using the result from Theorem~\ref{thm:2} for \mlp23 on terrains, we get a
2-approximate min-link path in $O(n^{7+\eps})$ time (since the path can bend
anywhere on a triangle of the terrain, the features are the triangles and
intrinsic dimension $d=2$). In this section we show that a faster,
$O(n^4)$-time 2-approximation algorithm is possible. We also consider encoding
visibility between all points on a terrain (not just between features, as the
visibility map from Section~\ref{sec:apxs3d} does): we give an $O(n^4)$-size
data structure for that, which we call the terrain's \e{global visibility map},
and provide an example showing that the size of the structure is worst-case
optimal.

We start with connecting approximations of \mlp23 and \mlp13 on terrains. Let
\opt be an optimal solution in an instance of \mlp23, let $\opt_e$ be the
optimal solution to \mlp13 in the same instance, and let \apxe be the
2-approximate path for the \mlp13 version output by the algorithm in
Section~\ref{sec:apxs3d} (Theorem~\ref{thm:2}); abusing notation, let \opt,
$\opt_e$ and \apxe denote also the number of links in the paths. Clearly,
$\apxe\le2\opt_e$; what we show is that actually a stronger inequality holds
(the inequality is stronger since $\opt\le\opt_e$):\begin{lemma}
\label{lem:e}$\apxe\le2\opt$.
\end{lemma}

\begin{proof}
Consider some link $pq$ on optimal path \opt from $s$ to $t$. Draw a vertical plane through $p$ and $q$ and denote as $p'$ and $q'$ the uppermost intersections of this plane with the boundaries of the triangles containing $p$ and $q$ (refer to Fig.~\ref{fig:2-appx}). Then $p'$ and $q'$ see each other, and they lie on edges of the terrain.
\begin{figure}
\centering
\includegraphics{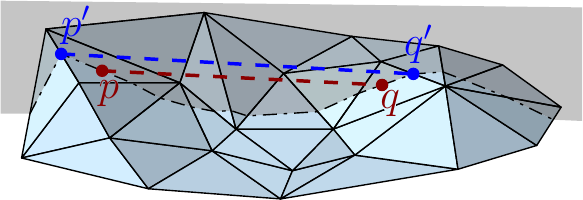}
\caption{For every pair of points $p \in f_p$ and $q \in f_q$ that can see each other, there exist points $p'$ and $q'$ on the edges bounding $f_p$ and $f_q$, respectively, that can also see each other.}
\label{fig:2-appx}
\end{figure}

Replace every link $pq$ of \opt by $p'q'$, and interconnect the consecutive
links by straight segments. Such interconnecting segments will belong to an
edge of the terrain, or go through the interior of a triangle containing the
corresponding vertex of the optimal path. The resulting chain of edges is a
proper path from $s$ to $t$ whose bends lie only on edges of the terrain. Thus,
it has a corresponding path in graph $G$ (refer to Theorem~\ref{thm:2}). The
length of such a path is at most $2\opt-1$, and it is not shorter than $\apxe$
(the shortest path in $G$). Therefore, $\apxe\le2\opt$.

\end{proof}

Lemma~\ref{lem:e} allows us to use the 2-approximation for \mlp13 as a
2-approximation for \mlp23. The former can be found more efficiently: by
Theorem~\ref{thm:2}, \apxe can be found in $O(n^4)$ time.

\begin{theorem}
  \label{thm:2tdr}
  A 2-approximation for \mlp23 in a terrain can be found in $O(n^4)$ time.
\end{theorem}

The running time of the algorithm in Theorem~\ref{thm:2tdr} is dominated by determining weak visibility between all $\binom{n}{2}$ pairs of edges; the approach from Section~\ref{sec:apxs3d} does it with brute force in $O(n^2)$ time per pair. An obvious question is whether this could be done faster for a single pair. We now show that this is hardly the case. We start from the analogous result for 2D polygonal domains:
\begin{theorem}
\label{thm:3sum2}
Determining weak visibility between a pair of edges in a polygonal domain with holes is 3SUM-hard.
\end{theorem}
\begin{proof}
The proof is by picture; see Fig.~\ref{fig:3sum}.
\end{proof}

\begin{figure}
\centering
\includegraphics{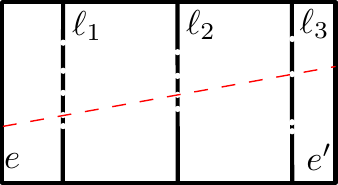}
\caption{Start from an instance of the 3SUM-hard problem GeomBase \cite{geombase}: Given a set $S$ of points lying on 3 parallel lines $\ell_1,\ell_2,\ell_3$, do there exist 3 points from $S$ lying on a line $\ell\notin\{\ell_1,\ell_2,\ell_3\}$? Construct an instance of the weak visibility problem for edges $e,e'$ in a polygonal domain: $\ell_1,\ell_2,\ell_3$ become obstacles and each point $p\in S$ is a gap punched in the obstacle; the lines are in a box whose two opposite edges (parallel to the lines) are the edges $e,e'$. The edges are weakly visible iff there exist 3 collinear gaps $p_i,i=1,2,3$, such that $p_i\in \ell_i$.}
\label{fig:3sum}
\end{figure}

The domain in Fig.~\ref{fig:3sum} can be turned into a terrain by erecting the lines $\ell_1,\ell_2,\ell_3$ into 3 vertical walls (the gaps in the lines become slits in the walls); similarly to the 2D case, the edges $e,e'$ weakly see each other iff GeomBase is feasible:
\begin{theorem}\label{thm:3sum3}
Determining weak visibility between a pair of edges in a terrain is 3SUM-hard.
\end{theorem}

The above 3SUM-hardness results are not the end of the story: the fact that determining weak visibility for a single pair of edges may require quadratic time does not imply that determining the visibility between all pairs of edges should require quatric time. In fact, the 3SUM-hardness of the 2D case (Theorem~\ref{thm:3sum2}) does not preclude existence of $O(n^2)$-time algorithm for finding \e{all pairs} of weakly visible edges in a polygonal domain with holes---such an algorithm is used, e.g., in Section~4 of \cite{diffuse}. Moreover, in \cite{alenex} it is shown that a data structure of $O(n^2)$ size can be built in $O(n^2)$ time, encoding visibility between \e{all pairs of points} in a domain; the data structure, which can be called the \e{global} visibility map of the domain, is an extension of the standard visibility graph that encodes visibility only between the domain's vertices. An immediate question is whether such a data structure can be built for terrains; below is our answer.

The \e{global visibility map} that encodes all mutually visible pairs of points on a terrain (or in a full 3D domain) will live in four dimensions---this is because a line in $\R^3$ has four degrees of freedom, and our data structure will use the projective dual $4D$ space $\Sp_{d}$ to the primary 3D space $\Sp_{p}$ where the terrain is located. A line $\ell\in\Sp_{p}$ will correspond to a point $\ell^{*}\in\Sp_{d}$. To build the global visibility map, consider a $5D$ space $\Sp_{5}$ where $\Sp_{p}$ and $\Sp_{d}$ are subspaces, and a point $O$ in $\Sp_{5}$ with coordinates $(0,0,0,0,1)$. The dual point $\ell^{*}\in\Sp_{d}$ for a line $\ell\in\Sp_{p}$ is constructed as follows: Draw a $4D$ hyperplane in $\Sp_{5}$ that goes through line $\ell$ and point $O$. A perpendicular line to such hyperplane that goes through $O$ intersects $\Sp_{d}$ in a point. This point will be $\ell^{*}$---the dual point to line $\ell$.

Now, the visibility map is a partition of $\Sp_{d}$ into cells, such that each cell contains points whose duals have the same combinatorial structure, i.e., they intersect the same set of obstacles' faces in $\Sp_{p}$.

\begin{lemma}\label{lem:mapUB}
The global visibility map that encodes all pairs of mutually visible points on terrain $T$ (or on a set of obstacles $\mathcal{O}$ in full 3D model) has complexity $O(n^4)$.
\end{lemma}
\begin{proof}
Let $\Li$ be a set of $n$ lines in $\Sp_{p}$. $\Li$ implies a subdivision $W$ of space $\Sp_{d}$ into cells that correspond to lines that touch the same sets of lines in $\Li$. $W$ consists of $0$-cells (vertices), $1$-cells (edges), $2$-cells, $3$-cells, and $4$-cells. The $k$-cells of $W$ correspond to a set of lines that intersects exactly $4-k$ lines of $\Li$. There are clearly $O(n^4)$ $0$-cells, since there are $n$ lines in $\Li$. For each $k$-cell, the number of incident $(k+1)$-cells is $O(1)$, since they correspond to the sets of lines we get by dropping incidence to $1$ of the $4-k$ lines (and $4-k$ is constant). Therefore, the number of $k$-cells is also bounded by $O(n^4)$ for all $k$. Hence, $W$ has complexity $O(n^4)$.

Now, consider our terrain $T$ (or a set of obstacles $\mathcal{O}$ in full 3D model) in $\Sp_{p}$. We are interested in the subdivision $S$ of $\Sp_{d}$ into cells that correspond to line segments that are combinatorially equal (their end points are on the same features of $T$ or $\mathcal{O}$). Then, $W$ is a sub-subdivision of $S$ (in the sense of subgraph, so something with fewer components). Hence, $S$ also has complexity $O(n^4)$.
\end{proof}

\noindent{\bf Remark. } The first part of the above argument (the complexity of configuration space of lines among lines in $3$-space) is a natural question and it is well-studied. McKenna and O'Rourke~\cite{McKenna1988} argue quartic bounds on the numbers of $0$-faces, $1$-faces and $4$-faces (although many proofs in their paper are omitted). They also describe how to compute the complex consisting of all $0$-faces and $1$-faces in $O(n^4 \alpha(n))$ time.

\begin{figure}
\centering
\includegraphics[width=0.7\columnwidth]{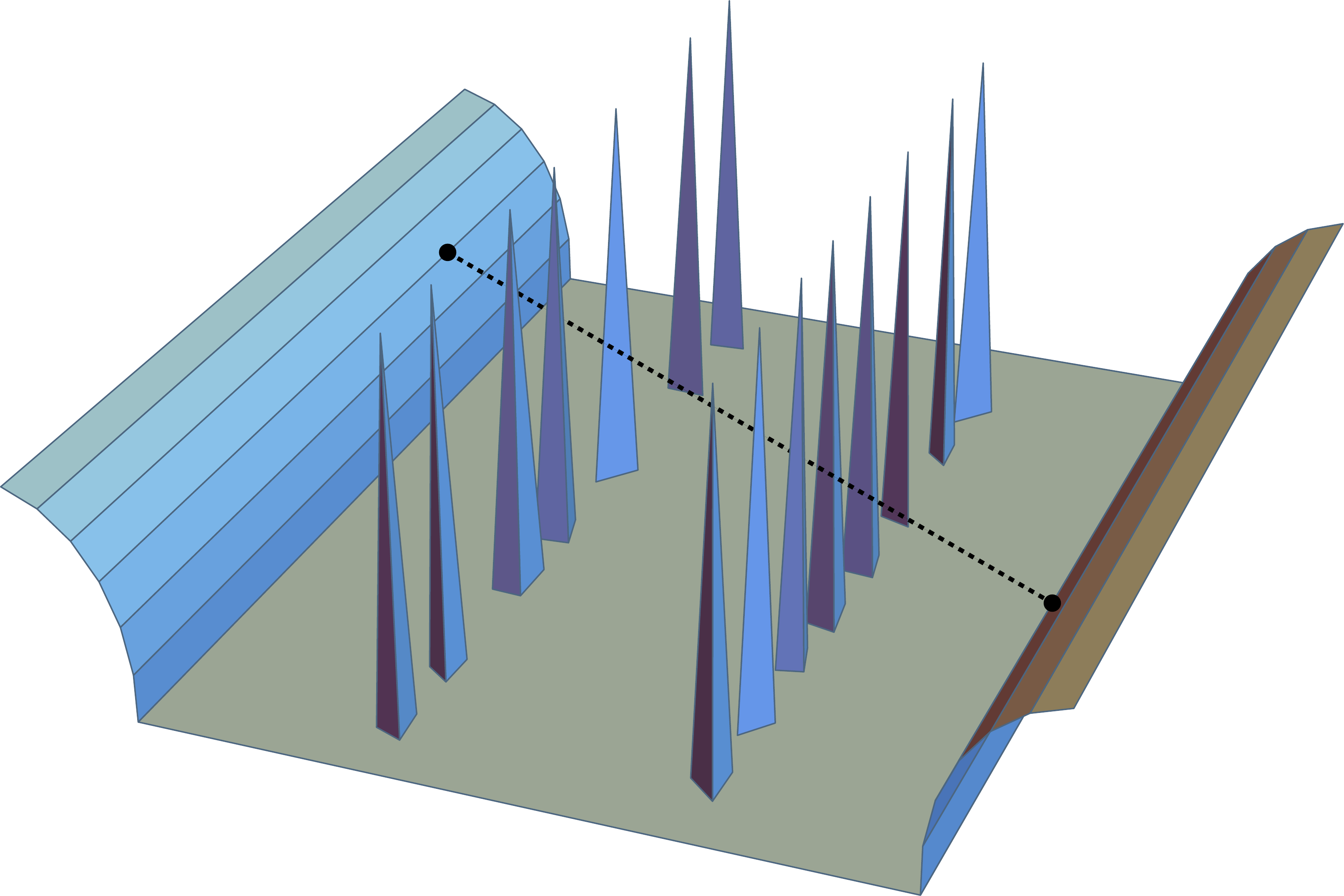}
\caption{Every vertex (0-face) in the visibility map corresponds to a line that crosses $4$ edges of the terrain. In this example, there is a line that connects any horizontal edge on the left-hand side with any horizontal segment on the right-hand side, and that also pins two spikes in the middle. Thus, there are $\Omega(n^{4})$ $0$-faces in the visibility map.}
\label{fig:cake}
\end{figure}

We now argue that the bound in Lemma~\ref{lem:mapUB} is tight: the global visibility map may have complexity $\Omega(n^4)$. (Other then possibly being an interesting result by itself,) this implies, in particular, that the running time of the algorithm in Theorem~\ref{thm:2tdr} may not be improved if one were to compute the weak visibility between all pairs of edges.
\begin{lemma}\label{lem:mapLB}
The global visibility map that encodes all pairs of mutually visible points on terrain $T$ can have complexity $\Omega(n^4)$.
\end{lemma}
\begin{proof}
  See Fig.~\ref{fig:cake}. It is easy to see that this construction yields a
  visibility map of complexity $\Omega(n^4)$.
\end{proof}
%
%
Lemmas~\ref{lem:mapUB} and~\ref{lem:mapLB} give tight bounds on the complexity of the visibility map:
\begin{theorem}The complexity of global visibility map, encoding all pairs of mutually visible points on terrain (or on a set of obstacles in 3D) of complexity $n$, is $\Theta(n^4)$.\end{theorem}

\section{Lower bound constructions for \mlp03 on terrains}
\label{ap:ve}

\begin{figure}[t]
\centering
\includegraphics[width=0.6\columnwidth]{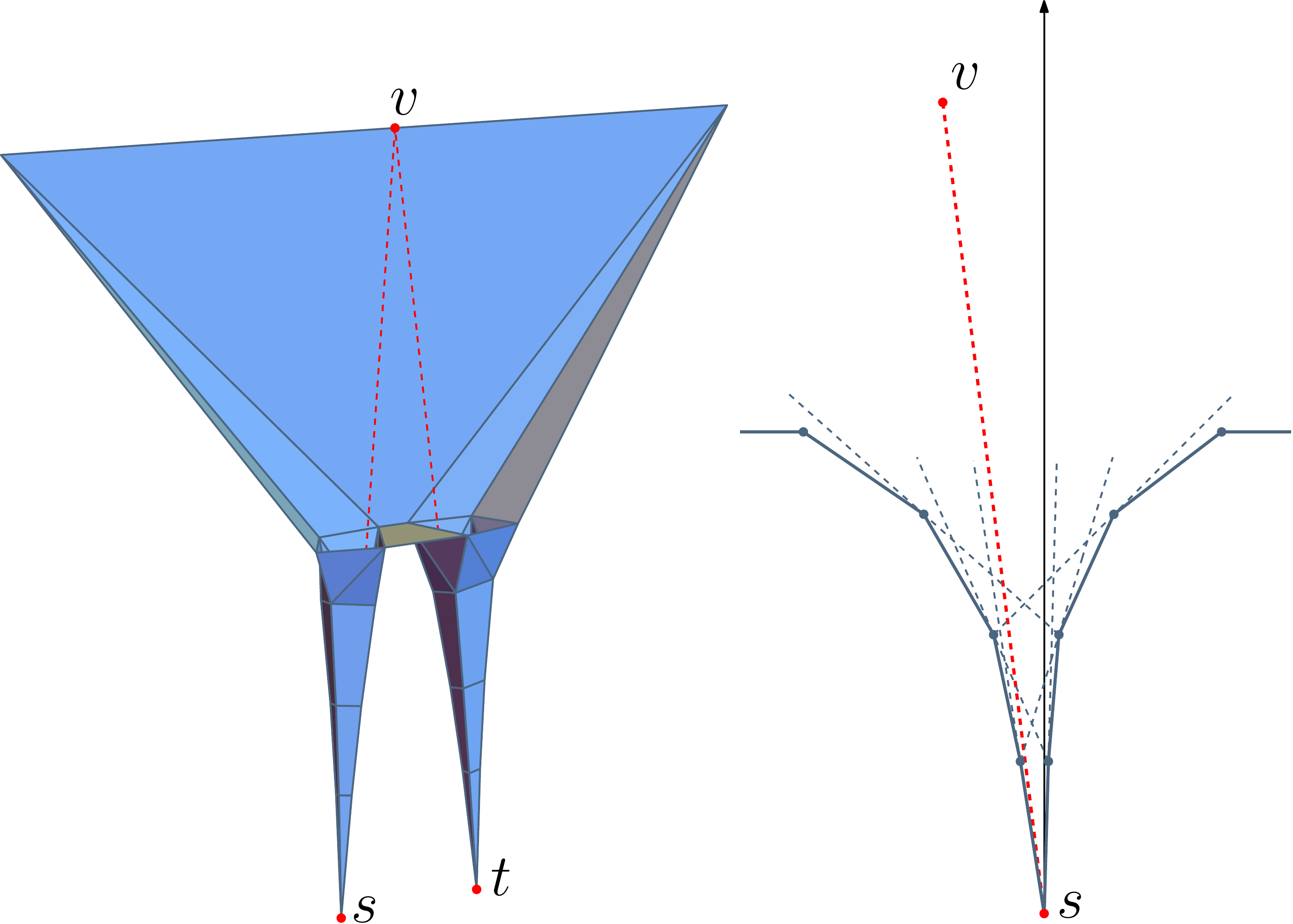}
\caption{On terrains, \mlp03 may be a factor $\Omega(n)$ worse than \mlp13. In the
  example shown, a solution to \mlp03 requires $\Omega(n)$ links, whereas \mlp13
  requires only two. Left: The 3D view on the trenches. Point $v$ is seen from
  both $s$ and $t$. Right: A cross-section of a trench. Vertices of each layer
  $i$ cannot see any other vertex further than the next layer.}
\label{fig:ve}
\end{figure}

Consider a terrain with two deep trenches whose vertices come in layers, see
for example Fig.~\ref{fig:ve}. Vertices in layer $i$ are connected to the
vertices of the previous layer, $i-1$, and vertices of the next layer,
$i+1$. Moreover, the only vertices visible from vertices in layer $i$ are the
ones in layers $i-1$ and $i+1$. Place $s$ and $t$ at the bottom of the
trenches. An $s\textrm-t$ path restricted to bend only at vertices of the
terrain will have a bend at each layer of the trenches, and therefore will have
$\Omega(n)$ links.

Now, place a tall steep face right outside the trenches, such that no vertices
of it can bee seen from any of the vertices inside the trenches, but some of
its interior and boundary edge $e$ can be seen from both $s$ and $t$. Then, a
solution to \mlp13 may bend in a point $v$ on $e$, and thus reach $t$ with two
links. Thus:

\begin{proposition}
  There is a terrain \T of $n$ vertices, and two vertices $s$ and $t$ on \T,
  such that a solution to \mlp03 requires $\Omega(n)$ times as much links as a
  solution to \mlp13.
\end{proposition}

\old{

\section{Min-link path through towers}
\label{sec:h}
\maarten {Remove this section for SoCG? It doesn't fit so well with the others.}
\frank{If we decide to keep it: rename the problem instances.}

In this section we address one issue with practical applicability of our main version, Terrain-\dr, by allowing the path bends to occur at a given elevation $h>0$ above the terrain; this models the possibility to install relays for wireless communication on height-$h$ towers (the vanilla variant of Terrain-\dr implicitly assumed $h=0$ which is not very practical). We denote the new problem by Terrain-\drs-$h$ (so Terrain-\drs = Terrain-\drs-0).

Say that edges $e,e'$ of the terrain are \e{$h$-visible} if they are weakly visible when both are shifted up by $h$, or equivalently, if there exist points $p\in e,p'\in e'$ such that the tips of height-$h$ towers installed at the points see each other (weakly visible edges are thus 0-visible). Determining the $h$-visibility for an arbitrary $h$ can be done in $O(n^2)$ time in the same way as determining the weak visibility (Lemma~\ref{lem:map}); the $h$-visibility graph can then be searched for shortest \s-\t path to get a path \apxe-$h$ that 2-approximates the min-link path whose bends stay at height $h$ over edges of the terrain. Finally, similarly to Terrain-\drs (Lemma~\ref{lem:e}), the path \apxe-$h$ is a 2-approximation also to the optimal path in Terrain-\drs-$h$ (in which the path is not restricted to bend over edges). It follows that Theorem~\ref{thm:2tdr} can be extended to arbitrary $h$:

\begin{theorem}
\label{thm:2tdrh}A 2-approximation for Terrain-\dr-$h$ can be found in $O(n^4)$ time.
\end{theorem}
Theorem~\ref{thm:2tdrh} allows us to find the approximate solution for a \e{given} $h$. One question of interest is how to compute the 2-approximations for \e{all} $h>0$. This may be important, e.g., when deciding what height towers to use: clearly, some heights are ``critical'' in that small change of the height leads to a change in the 2-approximate path computed by our algorithm, and it does not make sense to use towers of non-critical heights (for every non-critical height there exists a lower critical height giving path with the same number of links). The critical heights are a subset of those $h$ for which a pair of edges becomes $h$-visible (not all such heights are critical for the path: if the shortest path in the $h$-visibility graph does not use a newly created edge, then the $h$ is not critical; watching for changes in the graph's shortest path can be done by any semi-dynamic shortest path algorithm and is outside our scope). This motivates the \e{$h$-visibility problem}: Given two edges $e,e'$ of the terrain, find minimum $h$ for which they become $h$-visible. In the remainder of this section we give several results for the problem.

First of all, the 3SUM-hardness observations from the previous section (Theorem~\ref{thm:3sum3}) imply that it is 3SUM-hard to check whether the answer in the $h$-visibility problem is 0 or strictly positive:
\begin{corollary}\label{cor:hvis3sum}The $h$-visibility is 3SUM-hard, even to approximate to within an arbitrary factor.\end{corollary}

We now turn to (the more interesting) positive results. We can do binary search to approximate the answer to within an arbitrary factor; each step of the search is deciding whether a given $h$ suffices or not for the $h$-visibility:
\begin{proposition}A $(1+\eps)$-approximation to $h$-visibility can be found in $O(n^2\log\frac1\eps)$ time.\end{proposition}

Next, we show how to solve the $h$-visibility problem in $O(n^2 \log^2 n)$ time using parametric search~\cite{megiddo1983parametric}.

Let $\hat{e}$ and $\hat{e'}$ denote the edges $e$ and $e'$, lifted to some given height $h$. The (adapted) algorithm from Lemma~\ref{lem:map} allows us to test if $\hat{e}$ and $\hat{e'}$ are weakly visible from one another in $T_S = O(n^2)$ time. We now describe a parallel algorithm to test the weak visibility that runs in $T_P = O(\log n)$ time, using $P = O(n^2)$ processors. Plugging this into the parametric search framework will give us an algorithm with total running time $O(PT_P + T_PT_S \log P) = O(n^2 \log^2 n)$.

The global approach for our parallel algorithm is the same as that used in the sequential algorithm: we consider the parameter space $\hat e\times\hat{e'}$, in which each face of the terrain produces an obstacle of constant complexity. We build the arrangement of (the edges describing) these obstacles, and determine if there is a cell in the arrangement whose depth, that is, the number of obstacles containing it, is zero.

We use $O(n)$ processors---one for each face of the terrain---to compute the edges forming the obstacles in the parameter space $\hat{e} \times \hat{e'}$. We then compute the arrangement of all these edges in $O(\log n)$ time using $O(n^2)$ processors~\cite{anderson1996parallel}. The final step is to determine if there is a cell with depth zero. 
We do this as follows. We build a spanning tree on the dual graph of the arrangement, and construct an Euler tour $\E = c_1,..,c_k$ on this tree. Both these steps can be done in $O(\log n)$ time using $O(n^2)$ processors~\cite{shiloach1982parallelconnectivity,tarjan1985efficient}.

Let $d_1,..,d_k$ denote the depths of the cells in \E, let $\Delta_{ij} = d_j - d_i$ be the difference in depth between $c_j$ and $c_i$, and let $m_{ij}$ be the minimum depth of a cell in $c_i,..,c_j$, relative to the depth of $c_i$. That is, if $c_\ell$ is the cell with minimum depth in $c_i,..,c_j$ then $d_\ell = d_i + m_{ij}$. Clearly, $\Delta_{ii} = m_{ii} = 0$. We can easily compute $\Delta_{ij}$ and $m_{ij}$ from $\Delta_{i\ell}$, $\Delta_{(\ell+1)j}$, $m_{i\ell}$, and $m_{(\ell+1)j}$, for $i \leq \ell < j$, since going from $c_{\ell}$ to $c_{\ell+1}$ the depth changes by exactly one. Consider a balanced binary tree on \E. Each node $v_{ij}$ in this tree represents a subpath $c_i,..,c_j$. Using the procedure above in a bottom up fashion we can compute the values $m_{ij}$ and $\Delta_{ij}$ for all nodes in the tree in $O(\log n)$ time using $O(n^2)$ processors. The value $m_{1k}$ associated with the root of this tree, together with the depth at $c_1$ then gives us the global minimum depth of the tree. We can thus decide if the minimum is zero. So, all that remains is to compute the depth at $c_1$; this can easily be done in $O(\log n)$ time using $O(n)$ processors. We choose an arbitrary point in $c_1$. This point represents a line segment $s$ between a pair of points on $\hat{e}$ and $\hat{f}$. We then use one processor for each face of the terrain to test if it intersects $s$ in $O(1)$ time. Counting how many faces are intersected can again be done in $O(\log n)$ time using $O(n)$ processors.

The above procedure allows us to compute if $\hat{e}$ and $\hat{e'}$ are weakly visible. Plugging this into the parametric search framework gives us the following result.

\begin{theorem}
The $h$-visibility problem can be solved in $O(n^2\log^2 n)$ time.
\end{theorem}

We use the above procedure to compute the Pareto envelope of the (height, number of link) pairs; the envelope is the set of pairs $(h,l)$ such any $l$-link $s\textrm-t$ path must use towers of height at least $h$, and a path through towers of height $h$ must have at least $l$ links (i.e., one cannot simultaneously decrease both the tower height and the link distance). More specifically, in $O(n^4\log^2 n)$ time we compute the minimum heights at which all pairs of edges become weakly visible. We sort these $O(n^2)$ values $\{h_0,\dots,h_m\}$ by increasing height. Next, we build the weak-visibility graph $G_0$ for $h=h_0$. For each height $h_{i+1}$ (in sorted order) we can obtain the visibility graph $G_{i+1}$ from $G_i$ by inserting at most one edge into $G_i$. So, for each height $h_i$ we can compute the length $\ell_i$ of a shortest \s-\t path in graph $G_i$. This gives us $O(n^2)$ points $(h_i,\ell_i)$, on which we can compute the points on the Pareto envelope in $O(n \log n)$ time (or even $O(n)$ time, since the points can be reported by increasing height). The total running time for building the graphs is $O(n^2)$. Computing the shortest paths takes $O(n^4)$ time in total. This is dominated by the time it takes to compute the minimum heights. We conclude:
\begin{theorem}
We can compute a 2-approximation for Terrain-\drs-$h$ for all $O(n^2)$
critical heights in $O(n^4\log^2 n)$ time.
\end{theorem}

}

\section{Conclusion}We considered minimum-link paths in 3D, showing that most of the versions of the problem are hard but admit PTASes; we also obtained similar results for the diffuse reflection problem in 2D polygonal domains with holes. The biggest remaining open problem is whether pseudopolynomial-time algorithms are possible for the problems: our reductions are from 2-PARTITION, and hence do not show strong hardness. A related question is exploring bit complexity of the min-link paths in 3D (note that already in 2D simple polygons finding min-link path with integer vertices is weakly NP-hard \cite{ding}).

\paragraph*{Acknowledgments}
We thank Joe Mitchell and Jean Cardinal for
fruitful discussions on this work and the anonymous reviewers for their helpful comments.  M.L. and I.K. are supported by the
Netherlands Organisation for Scientific Research (NWO) under grants 639.021.123
and 639.023.208 respectively. V.P. is supported by grant 2014-03476 from the
Sweden's innovation agency VINNOVA. F.S. is supported by the Danish National
Research Foundation under grant nr.~DNRF84.

\frank{we have quite a few references, that take up quite a bit of space. Maybe
we can get rid of some of the less-related ones.}

\printbibliography

\old{
\clearpage
\appendix

\section{Lower bound constructions for \mlp03 on terrains}
\label{ap:ve}

Consider a terrain with two deep trenches whose vertices come in layers, see
for example Fig.~\ref{fig:ve}. Vertices in layer $i$ are connected to the
vertices of the previous layer, $i-1$, and vertices of the next layer,
$i+1$. Moreover, the only vertices visible from vertices in layer $i$ are the
ones in layers $i-1$ and $i+1$. Place $s$ and $t$ at the bottom of the
trenches. An $s\textrm-t$ path restricted to bend only at vertices of the
terrain will have a bend at each layer of the trenches, and therefore will have
$\Omega(n)$ links.

\begin{figure}[b!]
\centering
\includegraphics[width=0.6\columnwidth]{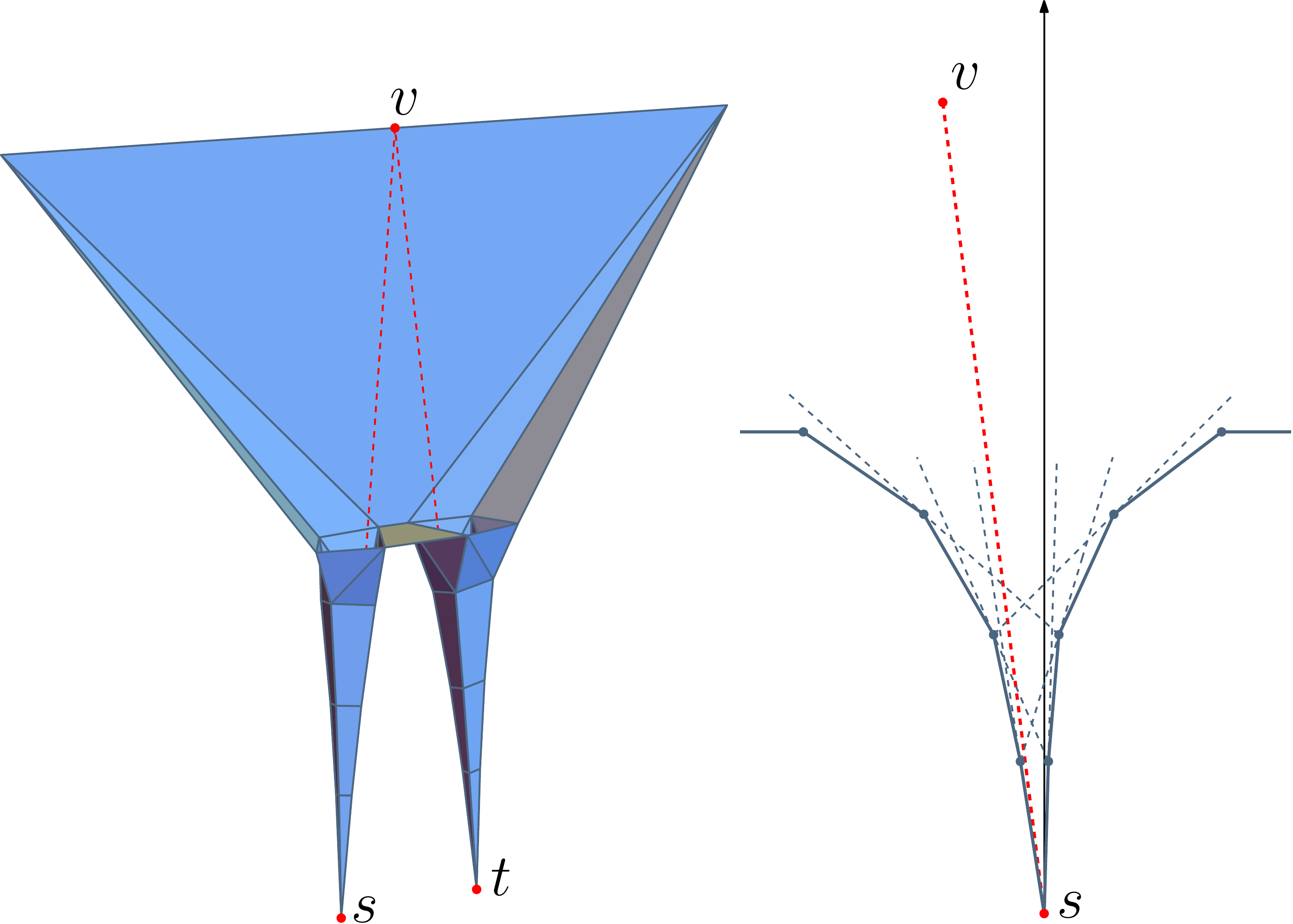}
\caption{On terrains, \mlp03 may be a factor $\Omega(n)$ worse than \mlp13. In the
  example shown, a solution to \mlp03 requires $\Omega(n)$ links, whereas \mlp13
  requires only two. Left: The 3D view on the trenches. Point $v$ is seen from
  both $s$ and $t$. Right: A cross-section of a trench. Vertices of each layer
  $i$ cannot see any other vertex further than the next layer.}
\label{fig:ve}
\end{figure}

Now, place a tall steep face right outside the trenches, such that no vertices
of it can bee seen from any of the vertices inside the trenches, but some of
its interior and boundary edge $e$ can be seen from both $s$ and $t$. Then, a
solution to \mlp13 may bend in a point $v$ on $e$, and thus reach $t$ with two
links. Thus:

\begin{proposition}
  There is a terrain \T of $n$ vertices, and two vertices $s$ and $t$ on \T,
  such that a solution to \mlp03 requires $\Omega(n)$ times as much links as a
  solution to \mlp13.
\end{proposition}

}

\end{document}

Do hardness proof with all details, e.g. start with a figure of a creased rectangle.
\paragraph{Features}
The only requirement is that all points on a feature see each other. The requirement is satisfied, e.g,,if the features are convex,.In particular, natural choices for the features are the tetrahedra in a tetrahedralization of a 3D domain in the free f, triangles in a triangulation of the terrain in the crawler version, edges in the-e version and in the 2d diffreffl.